\documentclass [12pt]{article}

\usepackage{amsmath,amsthm,amscd,amssymb}

\usepackage[small, centerlast, it]{caption}
\usepackage{epsfig}
\usepackage[titles]{tocloft}
\usepackage[latin1]{inputenc}
\usepackage{verbatim} 
\usepackage[active]{srcltx} 
\usepackage{fancyhdr}
\usepackage{caption}

\def\sp{\hskip -5pt}
\def\spa{\hskip -3pt}

\def\b1{{1\!\!1}}

\def\cB{{\ca B}}

\def\cF{{\ca F}}

\def\cH{{\ca H}}

\def\cL{{\ca L}}
\def\cM{{\ca M}}

\def\sK{{\mathsf K}}
\def\sH{{\mathsf H}}
\def\sP{{\mathsf P}}

\def\sS{{\mathsf S}}

\def\bC{{\mathbb C}}           

\def\bI{{\mathbb I}}

\def\bE{{\mathbb E}}

\def\bR{{\mathbb R}}

\def\bS{{\mathbb S}}

\def\gB{{\mathfrak B}}

\def\gO{{\mathfrak O}}

\def\beq{\begin{eqnarray}}
\def\eeq{\end{eqnarray}}

\newcommand{\ca}[1]{{\cal #1}}         

\usepackage{sectsty}
\sectionfont{\normalsize}
\subsectionfont{\normalsize\normalfont\itshape}
\newtheoremstyle{thm}
{12pt}
{12pt}
{\itshape}
{}
{\itshape\bfseries}
{}
{1em}
{}

\theoremstyle{thm}
\newtheorem{theorem}{Theorem}
\newtheorem{lemma}[theorem]{Lemma}
\newtheorem{proposition}[theorem]{Proposition}
\newtheorem{definition}[theorem]{Definition}

\newtheorem{remark}[theorem]{Remarks}

\begin{document}

\hfill{\sl  Revised version,  October  2015} 
\par 
\bigskip 
\par 
\rm


\par
\bigskip
\large
\noindent
{\bf Frame functions in finite-dimensional Quantum Mechanics and its Hamiltonian formulation on  complex projective spaces}
\bigskip
\par
\rm
\normalsize 


\noindent  {\bf Valter Moretti$^{a}$} and {\bf Davide Pastorello$^{b}$}\\
\par
\noindent Department of  Mathematics,  University of Trento, 
and  INFN-TIFPA, via Sommarive 14, I-38123 Povo (Trento), Italy.
\smallskip

\noindent $^a$ valter.moretti@unitn.it\qquad $^b$ d.pastorello@unitn.it\\
 \normalsize

\par

\rm\normalsize


\rm\normalsize


\par
\bigskip

\noindent
\small
{\bf Abstract}.  This work concerns some issues about  the interplay of standard and geometric (Hamiltonian) approaches to finite-dimensional quantum mechanics, formulated in the projective space. Our analysis relies upon  the notion and the properties of so-called frame functions,  introduced by A.M. Gleason to prove  his celebrated theorem. In particular,  the problem of associating quantum states with positive Liouville densities is tackled from an axiomatic point of view,  proving a theorem classifying all possible correspondences. A similar result is established for \emph{classical-like} observables (i.e. real scalar functions on the projective space) representing quantum ones. These correspondences turn out to be  encoded in a one-parameter class and, in both cases, the classical-like objects representing quantum ones result to be frame functions.
The requirements of $U(n)$ covariance and (convex) linearity  play a central r\^ole in the proof of those theorems. A new characterization of classical-like observables  describing quantum observables is presented, together with a geometric description of the $C^*$-algebra structure
of the set of quantum observables in terms of classical-like ones.  
\normalsize

\section{Introduction}
It has been known from a long time \cite{Kibble,AS,BH01} 
that quantum mechanics can be formulated 
as a  proper  Hamiltonian theory in the complex projective space. This observation, starting from \cite{Kibble}  even  produced some interesting attempts  towards non-linear extensions of quantum mechanics (see the last section of  \cite{BH01}
for references).
Actually, if the Hilbert space of the quantum formulation  is infinite-dimensional several  technical problems arise, especially related with the  notion of infinite dimensional manifold,  beyond the obvious fact that the physically relevant observables are unbounded self-adjoint operators. 
In this paper,  we only focus on the $n$-dimensional case, $n<+\infty$, whose interest is due to quantum information theory in particular.
Quantum-Hamiltonian formulation relies upon a  few ideas. First of all the space of phases  is chosen to be the complex projective space $\sP(\sH_n)$ constructed out of  the Hilbert  space $\sH_n$ of the considered quantum theory. The  manifold $\sP(\sH_n)$ possesses  a natural  {\em almost K\"ahler  structure}. That is a structure made of: (1) a symplectic form $\omega$, accompanied by (2) a   Riemannian metric  $g$, and  (3) an intertwining  almost complex structure $j$, transforming $g$ to $\omega$ and {\em viceversa}. The symplectic form $\omega$ permits  to introduce the standard  Hamiltonian formalism. As a second step, quantum observables, i.e. self adjoint operators $A$, are associated to real-valued functions $f_A: \sP(\sH_n)\to \bR$, thus  representing those operators in terms of classical-like observables. 
Fixing a Hamiltonian operator $H$ and its classically associated function $\cH := f_H$ (the classical-like Hamiltonian),  a point $p \in \sP(\sH_n)$  can be seen to evolve in time, $p=p(t)$, either 
 in accordance with an equation directly arising from   the Schr\"odinger equation  in $\sH$ (via  the quotient operation used to pass from $\sH_n$ to $\sP(\sH_n)$)
 or  in accordance with the Hamiltonian equations. The above mentioned correspondence of quantum and classical-like observables is fixed in such a way  that the two notions of evolution coincide.
It is worth stressing that the existence of such a  possibility is far from obvious.
Finally,  the Riemannian metric $g$ enters the picture in a nice way: Every notion of unitary 
time evolution corresponds to a notion of evolution along a suitable $g$-Killing flow.
In general, it turns out that  the unitary quantum symmetries  are represented by canonical  Hamiltonian symmetries, preserving  both the underlying Lie algebra structure and  the Riemannian metrical structures  of $\sP(\sH_n)$ (see Theorem \ref{teoQC}).
There are many other interesting features of this quantum-classical correspondence.  However, the agreement seems not to be guaranteed when comparing the two notions of   expectation value. 
Within the classical picture, where a physical system is described in a symplectic  space $(\cM,\omega)$, the expectation value $\bE_\rho(f)$ of a classical observable $f:\cM\rightarrow \bR$ is defined in terms of an integral with respect to a  Liouville measure  $\rho \:d\mu$.  Here, $\mu$ is the Liouville volume form constructed out of $\omega$ (and, for this reason it is invariant under symplectic diffeomorphisms), and the classical state $\rho$ is a  positive Liouville density satisfying the Liouville equation. Thus $\bE_\rho(f) := \int f \rho \: d\mu$. 
  Adopting the quantum framework, instead, an expectation value $\langle A \rangle_\sigma$ is nothing but  the trace of the product of the self-adjoint operators $A$, representing the considered observable,  and the density matrix $\sigma$, representing the state of the system. Thus  $\langle A \rangle_\sigma := \mbox{tr}(\sigma A)$. It is not clear from  the  literature what  is the most general way to define a correspondence from quantum to classical-like states (i.e.  Liouville densities on $\sP(\sH_n)$) such that quantum expectation values can be expressed as classical-like expectation values, i.e.,  $\bE_{\rho_\sigma}(f_A)= \langle A\rangle_\sigma$,  preserving the requirement  $\rho_\sigma  \geq 0$. If the last requirement does not hold, $\rho_\sigma$ cannot   be interpreted as a classical-like probability density. This issue is  tackled in this paper among others. In particular we  
classify all possible correspondences from quantum to classical-like states on the one hand, and  from quantum to classical-like observable on the other hand. In both cases, the classical-like objects representing quantum ones result to be frame functions (see below). These correspondences  
 fulfil  a list of natural requirements, $U(n)$ covariance and (convex) linearity in particular. We find that actually, there is room enough in the formalism to obtain positive Liouville densities $\rho_\sigma$, provided one drops another assumption  on observables that, however, does not seem so physically cogent. (Theorems \ref{main2}, \ref{main1bis} and \ref{teobound}). 
 As a byproduct, we also establish  another characterization of the small  class of functions on $\sP(\sH_n)$ describing quantum observables (Proposition \ref{propchat}). Eventually, a description of the unital $C^*$-algebra structure of the set of observables will be discussed   in terms of the geometric features of $\sP(\sH_n)$  (Theorem \ref{teolast}). \\

\begin{remark}  {\em It should be clear that a kind of  ``classical limit''  of a quantum system described
on a  finite dimensional Hilbert space should have a classical carrier space consisting of a
finite number of points. Therefore, for instance, the analysis in section 2.3 where quantum
and classical-like expectation values are discussed has nothing to do with the comparison
between the classical and quantum description of a given physical system since  no classical
limit is involved there, even if a geometrical classical-like description is introduced. The
physical systems studied in the present paper are of quantum nature and not classical.
For example, they admit a non-commutative algebra of observables as discussed in section
4.4, though these observables are presented into a classical-like picture because they are
described in terms of functions. However, the relevant product of observables is noncommutative
as is appropriate for quantum systems.} 
\end{remark}

\noindent The most important mathematical  notion we exploit  in our analysis is that of {\em frame function}.
It was introduced as a technical tool by A. M. Gleason to prove his celebrated theorem \cite{Gleason}.

\begin{definition}\label{ff}
  If $\sH$ is a separable complex Hilbert space, let
$\bS(\sH)$ denote the unit sphere centered on the origin.
$f:\bS(\sH)\rightarrow \bC$ is a {\bf frame function} on $\sH$ if  $W_f\in \bC$ exists  with:
\begin{equation}\label{weight}
\sum_{\phi\in N} f(\phi)=W_f \quad \mbox{for every Hilbertian basis $N$ of $\sH$.}
\end{equation}
\end{definition} 

\noindent The key step in the proof of Gleason's theorem is proving that the class of bounded real frame functions coincides 
to the class of  quadratic forms $f(\phi) = \langle \phi|A \phi \rangle$ where $A$ is any self-adjoint  trace-class operator on $\sH$.  A frame function on an {\em infinite} dimensional Hilbert space is always 
bounded, whereas in the finite-dimensional case  ($\mbox{dim}(\sH) \geq 3$), 
there exist infinitely many unbounded frame functions \cite{Dvu}. In \cite{DV1}, adopting a pure mathematical viewpoint,  we proved a proposition concerning sufficient conditions to assure that a frame function on a  complex finite-dimensional Hilbert space is  representable as a quadratic form without assuming the boundedness requirement {\em a priori}. Observing  that $\bS(\sH)$
 admits a unique positive regular Borel measure $\nu_{\sH}$ invariant under the left-action of unitary operators in $\sH$ and such that $\nu_{\cH}(\bS(\cH))=1$, we established the following theoretical result we will exploit in this work.

\begin{theorem}\label{ZAZ}
If $f: \bS(\sH)\to \bC$ is  a  frame function on a finite-dimensional complex Hilbert space $\sH$, with $\mbox{\em dim}(\sH) \geq 3$ and
$f \in \cL^2(\bS(\sH) ,d\nu_\sH)$, then
there is a unique linear operator $A : \sH \to \sH$ such that:
$
f(\psi)=\langle \psi|A\psi \rangle \quad \forall \psi\in \bS(\sH),
$
where $\langle \,\, |\,\,\rangle$ is the inner product in $\sH$. $A$ turns out to be Hermitean if  (and ony if) $f$ is real.
\end{theorem}    
\noindent Though that result is theoretically interesting,  nothing was said about the physical meaning of $\nu_{\sH}$ and the condition 
$f \in \cL^2(\bS(\sH) ,d\nu_\sH)$. This paper is also devoted to clarify these issues.
As a matter of fact, we will see that functions representing either quantum observables or  quantum states must be frame functions for $n>2$. It happens as a consequence of 
the (convex) linearity and $U(n)$-covariance of the maps associating classical-like objects to quantum ones.  Moreover $\nu_\sH$ will be established to be, up a positive factor,  just the the Liouville volume form necessary to compare expectation values.\\
The \emph{inverse quantization map} we introduce in section 4.1 gives a general way to construct classical-like observables and Liouville densities representing quantum objects, so there is an analogy with the \emph{quantizer-dequantizer scheme} within the tomographic approach to QM, where states are described as generalized Wigner functions.\\
This  paper is organized as follows. Sect.\ref{sect1} is devoted to summarize some  known aspects of the quantum-classical correspondence,  introducing some new ingredients and pointing out the issue of  the comparison of expectation values in Sect.\ref{secexp}.  In Sect.\ref{Sm}, always sticking  to the finite dimensional case,  we will present  some features of frame functions defined on the projective space, proving some fundamental theorems useful in the rest of the paper (Theorems \ref{teoDV2}, \ref{newtheorem} and \ref{teoFA}). In Sect.\ref{sectOFF}  we  will exploit the constructed formalism to deeply focus on the interplay of quantum observables/states  and classical-like observables/states, facing the problem of positivity of Liouville densities,   also discussing the construction  of  a suitable $C^*$-algebra structure on the space of frame functions  corresponding to  the analogous  structure of the space of quantum observables. The last section is dedicated to conclusions.

\section{Geometric Hamiltonian description of finite-dimensional quantum systems }\label{sect1}
\subsection{Elementary Quantum Mechanics}
In the absence of superselection rules, a quantum system is described in a complex Hilbert space $\sH$, whose elements determine the pure states of the system in the sense we discuss shortly.  With  that framework,  the self-adjoint elements of the $C^*$-algebra 
$\gB(\sH)$ of bounded operators on $\sH$  describe the observables which only  take bounded sets of values. (Unbounded observables can be constructed out of the bounded ones through limit procedures in the strong operator topology.)
  Two physically relevant two-sided $^*$-ideals $\gB_1(\sH)\subset \gB_2(\sH)$  of compact operators exist in $\gB(\sH)$. $\gB_2(\sH)$ is the space  of {\bf Hilbert Schmidt operators},
 $\gB_1(\sH)$ is the space of  {\bf trace-class  operators}.
Three  distances exist consequently.  (1)  $d_1(\cdot,\cdot)$ associated with  the norm 
$||A||_1 := \mbox{tr}(|A|)$ in the   Banach space
$\gB_1(\sH)$. (2)   $d_2(\cdot,\cdot)$ associated with the norm $||\cdot||_2$ induced by the scalar product on the Hilbert  space $\gB_2(\sH)$,
$(A|B)_2 = \mbox{tr}(A^*B)$  for $A,B \in \gB_2(\sH)$.
(3)  $d(\cdot,\cdot)$ associated with  the  standard $C^*$-algebra norm $||\cdot||$ on $\gB(\sH)$. An elementary computation proves that:

\begin{proposition}\label{teo2a}
If $\sH$ is a complex Hilbert space,  $||A|| \leq ||A||_2\leq ||A||_1$ when $A \in \gB_1(\sH)$.
\end{proposition}

\begin{definition} {\em If $\sH$ is a separable\footnote{Actually, separability is not strictly necessary from a mathematical viewpoint at least,  though this fact  would deserve a physical discussion.} complex Hilbert space with scalar product $\langle\cdot|\cdot \rangle$, $\sS(\sH)$ indicates the convex set of (quantum) {\bf  states} on $\sH$. They are  the  operators
$\sigma \in \gB_1(\sH)$ with  $\mbox{tr}(\sigma)=1$ which are positive (i.e $\langle \psi| \sigma \psi \rangle \geq 0$ if $\psi \in \sH$). 
{\bf Pure states} are the extremal points of $\sS(\sH)$, their set  is denoted by $\sS_p(\sH)$.  $\sigma \in \sS(\sH)$ is said {\bf mixed} when
$\sigma \not \in \sS_p(\sH)$.}
\end{definition}

\noindent    Pure states are related with the complex projective space on $\sH$. The (complex) {\bf projective space}   $\sP(\sH)$ over $\sH$ is the quotient $\sH/\spa\sim$ -- deprived  of $[0]$ -- where, for $\psi,\psi'\in \sH$,  $\psi \sim \psi'$ iff $\psi=\alpha\psi'$,  $\alpha \in \mathbb C\setminus\{0\}$.
$\bS(\sH)$ henceforth denotes the unit sphere in $\sH$ centred on the null vector.
With the topology induced by $\sH$, $\bS(\sH)$ is a connected Hausdorff  space, compact  only if $\mbox{dim}(\sH)< +\infty$. The  projection: $\pi : \bS(\sH) \ni \psi \mapsto [\psi] \in \sP(\sH)$
 is surjective, continuous and open  when equipping
$\sP(\sH)$ with the quotient topology.  $\sP(\sH)$ is  connected and  Hausdorff.
States enjoy the following properties, the proofs being well known  \cite{Mo}. From now on, $\mbox{sp}(A)$ denotes the spectrum of an operator $A$, in particular  $\mbox{sp}_p(A)$ and $\mbox{sp}_c(A)$ respectively indicate
the point spectrum and the continuous spectrum.
\begin{proposition}\label{teozero} If $\sH$ is a separable complex Hilbert space, the following facts hold.\\
{\bf (1)}   $\sS(\sH)$ and $\sS_p(\sH)$ are closed in  $\gB_1(\sH)$ and are  complete $d_1$-metric spaces.\\
{\bf (2)} If $\sigma\in \sS(\sH)$, then:
$\sigma^2 \leq \sigma  \quad \mbox{and}\quad \mbox{\em tr}(\sigma^2) \leq 1$, and the following facts are equivalent: 
(i)  $\sigma\in \sS_p(\sH)$; (ii) $\sigma^2= \sigma$; (iii) $\mbox{\em tr}(\sigma^2) =1$;
(iv)  $||\sigma||=1$;  (v) $||\sigma||_2=1$, (vi)
$\sigma = \psi\langle \psi| \cdot \rangle $ for some $\psi\in \bS(\sH)$.\\
{\bf (3)}  The homeomorphism exists  $\sP(\sH) \ni p \mapsto  \psi\langle \psi| \cdot \rangle \in \sS_p(\sH)$
for  $\psi\in \bS(\sH)$ with  $[\psi]=p$,  the topology assumed on $\sS_p(\sH)$ being   equivalently induced by $||\:||$ or $||\:||_1$ or $||\:||_2$, since 
$d_1(p,p')= 2d(p,p') = \sqrt{2}d_2(p,p')$ if $p,p' \in \sS_p(\sH)$.\\
{\bf (4)} If $\sigma \in \sS(\sH)$, then   $\mbox{\em sp}(\sigma) \setminus\{0\}\subset  \mbox{\em sp}_p(\sigma)$ is finite or countable with $0$ as uniquely possible  limit point.
If $q \in \mbox{\em sp}_p(\sigma)$ then $0 \leq q \leq 1$; the associated eigenspace $\sH_q$ has finite dimension if $q\neq 0$ and the sum of all eigenvalues, taking the geometric multiplicities into account, equals $1$.
If $K$ is a Hilbert basis of $Ker(\sigma)$ and  $\{\psi^{(q)}_i\}_{i=1,\ldots, \mbox{dim}(\sH_q)}$ 
a Hilbert basis of $\sH_q$, then 
  $K \cup\{\psi^{(q)}_i\:|\: i=1,\ldots, \mbox{\em dim}(\sH_q), q \in \mbox{sp}_p(\sigma)\}$ is a Hilbert basis of $\sH$.\\
{\bf (5)} Every  $\sigma \in \sS(\sH)$ is a finite or countable  convex combination of pure states, referring to  the  operator strong topology for infinite combinations. The spectral decomposition of $\sigma$ is an example of such convex decomposition.
\end{proposition}
\noindent From the physical side, for  a (bounded) observable $A = A^* \in \gB(\sH)$,  its spectrum $\mbox{sp}(A)\subset \bR$ is the set of its possible measured {\bf values}.
Moreover, 
 if  $\sigma \in \sS(\sH)$,  $E\subset  \mbox{sp}(A)$ is a Borel set  and $P_E$  is an orthogonal projector of the spectral measure of $A$, then  $\mbox{tr}(\sigma P_E)$ is the {\bf probability} of finding the  outcome of the measurement of $A$ in  $E$ when the state is $\sigma$. Correspondingly, $\mbox{tr}(\sigma A)$ is the {\bf expectation value} of $A$
in the state $\sigma$. Two observables $A,B$ are {\bf incompatible}, i.e., they cannot be measured simultaneously, if and only if $[A,B]\neq 0$. Kadison-Wigner  {\bf symmetries} are described by the action  of 
unitary or anti-unitary operators $U$ on states $\sigma$
like this:  $\sigma \mapsto  U \sigma U^*$. An one-parameter strongly continuous unitary groups $\{U_s\}_{s\in \bR}$ describes {\bf continuous symmetries}. Its  unique self-adjoint operator  $G$ (generally unbounded with domain $D(G)$) in the sense of Stone, i.e. $U_s = e^{-isG}$ for all $s\in \bR$, acquires a particular importance. When  $\{e^{-itH}\}_{t\in \bR}$ describes  time evolution of the system,  $H$ is the {\bf Hamiltonian} observable of the system.  If $\sigma(t) := e^{-i tH} \sigma e^{itH}$, the  {\bf Schr\"odinger equation} holds\footnote{Properly speaking,  Schr\"odinger equation arises from (\ref{Se}) when referring to normalized vectors $\psi(t)$ describing pure states $\psi(t)\langle \psi(t)| \cdot \rangle$, making  the simplest  choice of the arbitrary phase affecting $\psi(t)$.}: \beq \frac{d\sigma(t)}{dt} = -i[H, \sigma(t)]\:. \label{Se} \eeq
 The derivative refers to the  weak operator topology  
and the commutator is interpreted accordingly in $D(H)\times D(H)$. If, conversely, one  ascribes  time evolution to observables, keeping fixed the states, observables evolve as $A(t):= e^{itH} A e^{-itH}$ ({\bf Heisenberg picture}).

\subsection{Finite dimensional case: the geometric Hamiltonian picture}\label{secH} 
\begin{remark}   Throughout this paper
$\sH_n$ denotes a complex Hilbert space with finite dimension $n>1$ and $U(n)$
denotes the Lie group of unitary operators on $\sH_n$.
\end{remark}
\noindent When the dimension of the Hilbert  space $\sH_n$ is finite,  $\bS(\sH_n)$ and $\sP(\sH_n)$ become compact, second countable,  topological spaces. However the most interesting differences with respect to the infinite dimensional case concern the space of operators.  

\begin{proposition}\label{teoquat}
The following facts hold in $\sH_n$. \\
{\bf (1)}  The topologies of  $||\cdot||$, $||\cdot||_1$ and $ ||\cdot||_2$  on $\gB(\sH_n) =  \gB_2(\sH_n) =\gB_1(\sH_n)$ coincide. \\
{\bf (2)} $\sS(\sH_n)$ and  $\sS_p(\sH_n)$  are compact and,  if $\sigma \in \sS(\sH_n)$, the following inequalities hold:  $n^{-1/2} \leq ||\sigma||_2\leq 1$ and  $n^{-1}\leq ||\sigma|| \leq 1$. In both cases,
the least values of the norms are attained at   $\sigma= n^{-1}I$.\\
{\bf (3)}  Equip the set  $T$ of  operators $A = A^*\in \gB(\sH_n)$ such that  $\mbox{\em tr}(A)=1$  with the topology induced by $\gB(\sH_n)$.
As a subset of  the topological space $T$, $\sS(\sH_n)$ fulfils:
$$\partial \sS(\sH_n) = \{\sigma \in \sS(\sH_n) \:|\: \mbox{\em dim}(\mbox{\em Ran}(\sigma)) < n\}\:, \: \: \:
\mbox{\em  Int}(\sS(\sH_n)) = \{\sigma \in \sS(\sH_n) \:|\: \mbox{\em dim}(\mbox{\em Ran}(\sigma)) = n\}\:.$$
In particular:
 $\sS_p(\sH_n) = \{\sigma \in \sS(\sH_n) \:|\: \mbox{\em dim}(\mbox{\em Ran}(\sigma)) =1 \} \subset \partial \sS(\sH_n)$,
and   $\sS_p(\sH_n) = \partial \sS(\sH_n)$
if and only if $n=2$.
\end{proposition} 

\noindent Let us  identify $\sH_n$ with $\bC^n$ by choosing a Hilbert basis.
With this identification $\sH_n$ acquires the structure of a  {\em real $2n$-dimensional smooth manifold}. This structure  does not depend on the choice of the basis as one immediately proves.  That structure induces analogous structures on the topological  spaces $\bS(\sH_n)$ and  $\sP(\sH_n)$.
From now on we identify $\sP(\sH_n)$ to $\sS_p(\sH_n)$ in view of (3) in 
proposition \ref{teozero}. As a consequence, we can take advantage of  the transitive action of the compact Lie group $U(n)$ on $\sS_p(\sH_n) \equiv \sP(\sH_n)$:
\beq
U(n) \times \sP(\sH_n) \ni (U, p) \mapsto  \Phi_U(p) :=   U p U^{-1} \in  \sP(\sH_n)\:.\label{ta}
\eeq
 A sketch of proof of the following proposition is in the appendix. 

\begin{proposition} \label{profgeo} The following facts hold in the  real smooth manifold $\sH_n$.\\
{\bf (a)}  $\bS(\sH_n)$ is a  real $(2n-1)$-dimensional embedded submanifold of $\sH_n$.\\
{\bf (b)}  $\sP(\sH_n)$ can  be equipped with a real $(2n-2)$-dimensional smooth  manifold structure in a way such that both  the continuous  projection $\pi : \bS(\sH) \ni \psi \mapsto [\psi] \in \sP(\sH)$ is  a smooth submersion and the transitive action (\ref{ta}) is smooth.
\end{proposition}
\begin{remark}
Henceforth $iu(n)\subset \gB(\sH_n)$ -- where $u(n)$ is the Lie algebra of $U(n)$-- denotes the real space of self-adjoint operators.
\end{remark}
\noindent The following proposition, whose proof  is in the appendix, establishes a useful way to describe the tangent space $T_\sigma \sP(\sH_n)$.

\begin{proposition} \label{propvett}
The tangent vectors  $v$ at $p \in \sP(\sH_n)\equiv \sS_p(\sH_n)$ are all of the 
elements in $\gB(\sH_n)$ of the form:
$v = -i[A_v, p]$, for some  $A_v\in iu(n)$.
Consequently,  $A_1,A_2 \in iu(n)$ 
define the same vector in $T_p \sP(\sH_n)$  iff $[A_1-A_2, p]=0$. 
\end{proposition}

\noindent 
As is well known \cite{AS,CLM,BH01}, $\sP(\sH_n)$ has also a structure of a $2n-2$-dimensional  {\bf symplectic manifold}, where the symplectic form (a closed non-degenerate smooth $2$-form) is, for any fixed value of the constant $\kappa>0$:
\beq
\omega_p (u, v) := 
- i\kappa \:\mbox{tr}\left(p  [A_u, A_v]\right) \quad u,v  \in T_p \sP(\sH_n) \:.\label{omega}
\eeq
This definition is well-posed, since,  by direct inspection, one sees that the right-hand side of (\ref{omega}) is fixed if adding  to $A_u$ or $A_v$ operators commuting with $p$. The constant $\kappa>0$ is a natural degree of freedom  we do not fix at this stage. It is introduced just for future convenience. As we shall see shortly,  $\kappa$ affects  the form of the classical-like observables associated with the quantum ones (\ref{CQ}) and in the litterature 
$\kappa$ is usually assumed  to be either $1$ \cite{BH01} or $1/2$ \cite{BSS04}. \\
 The symplectic structure allows us to  take advantage of  the usual Hamiltonian machinery, whose relation  with quantum mechanics formalism will be examined shortly. Just we recall some general facts (so $\sP(\sH_n)$ can be replaced for any symplectic manifold).
A diffeomorphism $F : \sP(\sH_n) \to \sP(\sH_n)$ is said to be {\bf symplectic} iff it preserves the symplectic form:  $F_\star \omega = \omega$. 
 For every smooth  $f:  \sP(\sH_n) \to \bR$ one defines the associated {\bf Hamiltonian (vector) field}
$X_f$ as the unique vector field satisfying $\omega_p(X_f,  \cdot ) = df_p$.
When $\cH$ is the {\bf Hamiltonian function} (time-independent for the sake of simplicity) of a  physical system described on  $\sP(\sH_n)$, the integral curves of $X_\cH$,  the solutions of {\bf Hamilton equations}:
\beq
\frac{dp}{dt} = X_\cH(p(t))\:, \label{He}
\eeq
represent the time evolution of the system.  Generally speaking, the evolution along the  integral curves of $X_f$ (which, on  $\sP(\sH_n)$, are complete since it  is compact) defines a one-parameter group of symplectic diffemorphisms
called the {\bf Hamiltonian flow} generated by the smooth function  $f$.
The {\bf Poisson bracket} of  a pair of smooth functions $f, g : 
\sP(\sH_n) \to \bR$ is $\{f,g\}:= \omega(X_f,X_g)$ and the remarkable formula 
holds $[X_f, X_g] = X_{\{f,g\}}$,  the  former commutator  being the  Lie bracket of vector fields.\\
 Coming back to  $\sP(\sH_n)$ explicitly, as is known \cite{AS,BH01}, it also admits a  positive  preferred smooth metric.  Up to the  factor 
$2\kappa$,  it is the so-called {\bf Fubini-Study metric}\footnote{With the  suggestive notation $-i[A, p] = dp$, for $A=B$, the metric (\ref{fs}) assumes the more  popular form $ds^2 = g_p(dp, dp) = 2 \kappa \:\mbox{tr} \left(p (dp)^2 \right)$, that  is equivalent to (\ref{fs}) through the polarization identity.}:
\beq
g_p \left(u, v\right) = - \kappa \:\mbox{tr}\left(p  \left( [A_u, p]  [A_v, p] +  [A_v, p]  [A_u, p] \right)\right) \quad  u,v  \in T_p \sP(\sH_n) \label{fs}\:.
\eeq
Finally a  $\omega$-$g$-compatible {\bf almost complex structure} exists \cite{AS,BH01},  explicitly given by the class $j$ of  linear maps  \cite{GCM2005}:
\beq   j_p :  T_p \sP(\sH_n) \ni v \mapsto i[v, p]  \in T_p \sP(\sH_n)\:, \quad p\in \sP(\sH_n)\:.\eeq
Indeed $p \mapsto j_p$ is  smooth, fulfils $j_p j_p = -I$ and  $\omega_p(u, v) = g_p(u, j_p v)$ if $u,v \in T_p \sP(\sH_n)$. (Symmetry of $g$ and anti-symmetry of $\omega$ also imply
$\omega(u, jv)= -g(u,v)$,  $g(ju,jv)= g(u,v)$ and $\omega(ju,jv) = \omega(u,v)$.)
$(\omega, g, j)$ is an {\bf almost K\"ahler structure} on $\sP(\sH_n)$.\\
Let us come to the interplay of Hamiltonian and Quantum formalism  \cite{AS,BH01}. It  relies upon the idea to associate  a quantum observable $A \in iu(n)$ to a classical-like observable:
\beq f_A : \sP(\sH_n) \ni p \mapsto  \kappa \:\mbox{tr}(p A) + c \:\mbox{tr} A\in   \bR\:. \label{CQ}\eeq
The constant $c\in \bR$ can be fixed arbitrarily. Once again  $c$ is another natural degree of freedom, allowed since  it does not affect the known results we are about stating. The core of the Hamiltonian description of quantum physics is stated in the following theorem proved  in the appendix. Other interesting aspects exist, related with, for instance, submanifolds with fixed energy,  geometric description of quantum entanglement,  theory of integrable systems etc. We omit them  for the shake of brevity (see \cite{AS,BH01,BSS04,GCM2005}).

\begin{theorem} \label{teoQC} Consider a quantum system described on $\sH_n$. Equipp $\sP(\sH_n)$ with the triple  $(\omega, g, j)$ as before. For every $A \in iu(n)$, define the function $f_A : \sP(\sH_n) \to \bR$ as in  (\ref{CQ}). Then the Hamiltonian field associated with $f_A$ reads:
\beq X_{f_A}(p) = -i[A,p] \quad \mbox{for all}\:\:  p\in \sP(\sH_n)\:,\label{XfA}\eeq
and the  following facts hold.\\
{\bf (a)}  $\bR \ni t \mapsto p(t) \in \sS_p(\sH_n)$ is the evolution of a pure quantum state fulfilling Schr\"odinger equation (\ref{Se}) with Hamiltonian $H\in  iu(n)$ if and only if $\bR \ni t \mapsto p(t) \in \sP(\sH_n)$
satisfies Hamilton equations (\ref{He}) with  Hamiltonian function $\cH := f_H$.\\
Similarly,  Hamiltonian evolution of classical-like observables is equivalent the Heisenberg evolution of corresponding quantum observables: 
$f_A(p(t)) = f_{e^{itH}Ae^{-itH}}(p)$.\\
{\bf (b)} If $A, H \in iu(n)$, then:
\beq \{f_A, f_H\} = f_{-i[A,H]}\:.\label{ppc}\eeq So in particular $A$ is a quantum constant of motion 
if and only if $f_A$ is a classical-like constant of motion when $\cH = f_H$ is the Hamiltonian function.\\
{\bf (c)}  If $U\in U(n)$ the map $\Phi_U: \sP(\sH_n) \to \sP(\sH_n)$ as in (\ref{ta}),
describing the action of the quantum symmetry $U$ on states, is both a symplectic diffeomorphism and an isometry of $\sP(\sH_n)$ and  thus $X_{f_A}$ is a $g$-Killing fields for every $A\in iu(n)$. Finally the covariance relation holds:  
$$f_{A}(\Phi_U(p))= f_{U^{-1}AU}(p)\quad \mbox{for all $A\in iu(n)$, $p\in \sP(\sH_n)$, and  $U\in U(n)$.}$$
\end{theorem}

\begin{remark} {\em $\null$\\
{\bf (1)} Changing the form $c\:\mbox{tr}(A)$ of the constant term in the right-hand side in (\ref{CQ}) only affects the validity of (\ref{ppc})  in the thesis of thm \ref{teoQC}.\\
{\bf (2)} It is possible to prove that, remarkably,  a $g$-Killing field is necessarily a Hamiltonian field $X_{f_A}$ for some $A\in iu(n)$ \cite{BH01}.}
\end{remark}
\subsection{Matching quantum and classical-like expectation values} \label{secexp}
The appearance of the  constants $\kappa$ and $c$ in (\ref{CQ}), though maybe unusual,  does not give rise to any problem 
in comparing  quantum dynamics with  Hamiltonian one
and in discussing the interplay of classical-like and quantum symmetries,  as done in theorem \ref{teoQC}.  However one may ask if  further  degrees of freedom can be found  preserving the nice agreement  of quantum and classical dynamics. This problem will be tackled shortly, proving that the answer is negative. The said issue  is actually related with another problem we go to introduce. If we take seriously the fact that $A$ and $f_A$ are quantum and classical-like  observables associated to each other,  we have to be more precise on how the values obtained by measurements  of these observables are related.  To focus on this relation we have to compare quantum and classical-like expectation values, referred to corresponding states. 
 In Hamiltonian mechanics, referring to a {\bf statistical state} $\rho$ described in terms of a {\bf Liouville density}, the expectation value is  the integral of the product of $\rho$ and the considered observable $f$ with respect to the {\bf Liouville} (positive Borel) {\bf measure} $m:= \omega \wedge \cdots\mbox{($n$ times)}\cdots \wedge \omega$, where  $2n$
is the dimension of the symplectic  space: 
\beq \bE_\rho(f) := \int_{\sP(\sH_n)} \rho(p) f(p) \:\: d m(p)\:. \label{inte}\eeq
Here, {\bf Liouville  (probability) densities} are   non-negative functions $\sP(\sH_n) \ni p \mapsto \rho(p) \in \bR$ with\footnote{Actually,  in our case, $\cL^1(\sP(\sH_n), m)\subset 
 \cL^2(\sP(\sH_n), m)$ since $m(\sP(\sH_n))$ is finite, $\sP(\sH_n)$ being compact.}  $\rho \in \cL^1(\sP(\sH_n), m)\cap
 \cL^2(\sP(\sH_n), m)$,  satisfying $||\rho||_{\cL^1(m)}=1$.
The classical-like observables $f$ are supposed to fulfil  $f \in \cL^2(\sP(\sH_n),m)$ so that (\ref{inte}) makes sense.
Liouville densities evolve in time satisfying the celebrated {\bf Liouville equation}.
{\bf Sharp states} $\rho_{p_0}$ defined by a single point $p_0 \in \sS(\sH_n)$ can be thought of as Dirac's measures $\mu_{p_0}$ on the Borel $\sigma$-algebra
on $\sP(\sH_n)$  concentrated on $p_0$. The expectation value of an observable $f$ therefore coincides to its  evaluation at $p_0$:
\beq \bE_{\rho_{p_0}}(f) := \int_{\sP(\sH_n)}  f(p) \:\: d \mu_{p_0}(p) = f(p_0) \:.\eeq
In quantum mechanics  we  have mixed and pure states, respectively resembling statistical and sharp classical states. These states evolve in accordance 
with  Schr\"odinger equation, i.e.,  by means of the unitary evolutor associated with the Hamiltonian operator.
The expectation value of a quantum observable $A\in iu(n)$ referred to a state $\sigma$ is:
\beq
\langle A \rangle_\sigma := \mbox{tr}(\sigma A)\:. \label{tra}
\eeq
Comparing classical-like and quantum observables, we aspect that a natural requirement that could help fix $\kappa$ and $c$ is  a constraint like this:
\beq \mbox{tr}(A \sigma) = \int_{\sP(\sH_n)} f_A(p) \rho_\sigma(p) dm(p) \label{trAs0}
\quad \mbox{for every $A\in iu(n)$ and $\sigma\in \sS(\sH_n)$,}\eeq
$\rho_\sigma$ is a Liouville density associated with $\sigma$
through some unknown procedure.
In  \cite{Gibbons},  Gibbons  proved that there is a way to associate quantum and classical-like states such that,  if the classical-like observable $f_A$ with $\kappa=1$ and $c=0$ corresponds to  $A\in iu(n)$, then (\ref{trAs0}) holds true. It happens for a measure  $\mu$, in place of $m$, related with the Fubini-Study metric ($\mu = n(n+1) \mu_n$, where $\mu_n$ is that defined in proposition \ref{propmu} below) and  for:
\beq
\rho_\sigma(p) :=  \mbox{tr}(\sigma p) - \frac{1}{n+1}\:. \label{Gibbonsf}
\eeq
With the given definitions one easily sees that the Schr\"odenger evolution of the quantum states is equivalent to the evolution along the Liouville equation 
for the associated classical-like states. Nevertheless, the evident problem is that $\rho_\sigma \geq 0$  is not generally  true, so $\rho_\sigma$ {\em cannot} define a probability density.  However, as  stressed in  \cite{Gibbons}, since (\ref{trAs0}) is valid, one cannot produce non-physical results (e.g. $f_A \geq 0$ but $\bE_{\rho_\sigma}(f_A)<0$) dealing with  the few functions of the form  $f_A$, varying  $A\in iu(n)$.
 Nevertheless, there is no way to think of $\rho_\sigma$ as a classical-like state  when dealing with general classical-like observables $f :  \sP(\sH_n) \to \bR$.\\ 
In the rest of the paper, we wish to  focus on the interplay of quantum and classical description, studying all possible correspondences from quantum states to Liouville densities on $\sP(\sH_n)$ satisfying natural requirements.  These requirements in particular, fix a relation between $\kappa$ and $c$. We also  establish that both classical-like observables representing quantum ones and densities representing quantum states {\em must} be frame functions.  We will also prove  that, in the found picture, 
 there is room enough
to fix the positivity problem of $\rho_\sigma$,  preserving the validity of  theorem \ref {teoQC}.

\section{More about frame functions on the  projective space}\label{Sm}
To tackle the issues  illustrated  in Sect.\ref{secexp},  we  introduce some new preparatory results about frame functions.
\subsection{Frame functions on $\sP(\sH_n)$}
First of all we need to restate the definition of frame function and theorem \ref{ZAZ} on $\sP(\sH_n)$  rather than on  $\bS(\sH_n)$. To this end, we exploit the existence of a suitable measure on $\sP(\sH_n)$ induced by the measure $\nu_n$ defined on $\bS(\sH_n)$ mentioned in theorem \ref{ZAZ} and therein denoted by $\nu_{\sH}$.
As  $\bS(\sH_n)$ is homeomorphic to the quotient of compact groups $U(n)/U(n-1)$, it is endowed with a  $U(n)$-left-invariant  regular positive Borel measure, $\nu_n$, that is uniquely determined by its  normalization $\nu_n(\bS(\sH_n)) =1$ (see \cite{Mackey} and Chapter 4 of \cite{BR}). 
We have the following proposition whose proof is in the appendix.

\begin{proposition}\label{propmu} Let  $\nu_n: \bS(\sH_n) \to [0,1]$ denote the unique  $U(n)$-left-invariant regular Borel measure   with $\nu_n( \bS(\sH_n))=1$.
There exists a unique positive Borel measure $\mu_n$ over $\sP(\sH_n)$ such that, if $\pi  : \bS(\sH_n) \to \sP(\sH_n)$ is the natural projection map, then:
\beq
\mbox{$f\circ \pi   \in \cL^1(\bS(\sH_n), \nu_n)$} \quad\mbox{if}
\quad \mbox{$f \in \cL^1(\sP(\sH_n), \mu_n)$,}\quad \mbox{and}\quad 
\int_{\sP(\sH_n)} f d\mu_n  = \int_{\bS(\sH_n)} f \circ \pi  \:\: d\nu_n\:. \nonumber
\eeq
The measure $\mu_n$ fulfils the following.\\
{\bf (a)}  Referring to the smooth action (\ref{ta}), $\mu_n$ is the unique $U(n)$-left-invariant regular Borel measure on $\sP(\sH_n)$ with $\mu_n(\sP(\sH_n))=1$.\\
{\bf (b)}   It coincides to the Liouville volume form  induced by
$\omega$ up to  its normalization.\\
{\bf (c)}  It coincides to the  Riemannian measure induced by  $g$ up to  its normalization.
\end{proposition}
\noindent A  frame function as in definition  \ref{ff} actually determines  a function on $\sP(\sH_n)$.
This is because both the unit vectors $\psi$  and $\alpha \psi$, for $|\alpha|=1$,  can be completed to a Hilbert basis of $\sH$ by adding {\em the same set} of $n-1$ vectors $\psi_2,\psi_3,\ldots$. Requirement (\ref{weight}) 
for a frame function $f: \bS(\sH) \to \bC$
therefore  implies: $f(\psi) = W_f - \sum_{i \geq 2} f(\psi_i) = f(\alpha \psi)$.
We  may consequently  state an equivalent  definition of frame function on $\sP(\sH)$. The only problem concerns the analogue of the  notion of Hilbert basis stated on $\sP(\sH_n)$ instead of $\sH_n$.
We have the following helpful elementary result with  $\pi  : \bS(\sH_n) \to \sP(\sH_n)$ as before.
\begin{proposition}\label{propBasis}
Let  $\sH$ be a separable complex Hilbert space.  $N \subset \sP(\sH)$ can be written 
as $\{\pi(\psi)\}_{\psi \in M}$ for some Hilbertian basis  $M \subset \sH_n$
if and only if both $d_2\left(p,p'\right)= \sqrt{2}$ for $p,p' \in N$ when $p\neq p'$ and $N$ is maximal with respect to 
this property.
\end{proposition}
\begin{proof}
As
$
d_2\left(\psi\langle\psi| \:\cdot\: \rangle,  \phi\langle\phi| \:\cdot\: \rangle \right) = \sqrt{||\psi||^4
+||\phi||^4 - 2|\langle \psi| \phi \rangle|^2}$ if $\psi,\phi \in \sH$,
 for
 $\psi,\phi \in \bS(\sH)$, one has
$d_2\left(\psi\langle\psi| \:\cdot\: \rangle,  \phi\langle\phi| \:\cdot\: \rangle \right)=\sqrt{2}$ if and only if 
$\psi \perp \phi$.
The proof concludes noticing that the maximality property  in the thesis is equivalent to that of a  Hilbertian basis. \end{proof}

\noindent We remark the fact  that 
$\sqrt{2}= \max\{d_2\left(p,p'\right)\:|\: p,p' \in  \sP(\sH)\}$
 for every separable complex Hilbert space $\sH$. Moreover, if $\mbox{dim}(\sH)=n < +\infty$ the maximality condition is equivalent
to say that $N$ contains exactly $n$ elements.
\begin{definition}
{\em If $\sH$ is a separable complex Hilbert space,  $N \subset \sP(\sH)$
is  a {\bf basis} of $\sP(\sH)$ if $d_2\left(p,p'\right)= \sqrt{2}$ for $p,p' \in N$ with $p\neq p'$ and $N$ is maximal with respect to  this property.}
\end{definition}

\noindent We may give a definition of frame function on $\sP(\sH_n)$ equivalent to that in definition \ref{ff}.
\begin{definition}{\em 
A map $F:\sP(\sH_n)\rightarrow \bC$ is a {\bf frame function}  if there is 
 $W_F\in \bC$ with:
\begin{equation}\label{weight2}
\sum_{i=1}^n F(x_i)=W_f \quad \mbox{for every basis $\{x_i\}_{i=1,\ldots, n}$ of  $\sP(\sH_n)$.}
\end{equation}}
\end{definition}

\noindent Theorem \ref{ZAZ} can now be restated referring to the measure $\mu_n$ and  completing the statement  by adding some other elementary facts. The non elementary result is (b).
\begin{theorem}\label{teoDV2}
In  $\sH_n$  the following  holds.\\
{\bf (a)} If $A \in \gB(A)$ then
 \beq F_A(p) := \mbox{tr}\left( pA\right)  \quad \mbox{for $p \in \sP(\sH_n)$.}\label{FA}\eeq
defines a frame function with $W_{F_A}= \mbox{\em tr} A$  which belongs to $ \cL^2(\sP(\sH_n), d\mu_n)$. \\
{\bf (b)} If $F : \sP(\sH_n) \to \bC$  is a frame function, $n>2$ and $F \in  \cL^2(\sP(\sH_n), d\mu_n)$, 
then there is a unique $A\in \gB(\sH_n)$ such that  $F_A = F$.\\
{\bf (c)} Defining the subspace, closed if $n>2$:
$$\cF^2(\sH_n) := \{  F: \sP(\sH_n) \to \bC \:|\: F \in \cL^2(\sP(\sH_n), d\mu_n) \quad \mbox{and $F$ is a frame function}\}\:,$$
 $ M : \gB(\sH) \ni A \mapsto F_A \in  \cF^2(\sH_n)$
is a complex vector space injective homomorphism, surjective if $n>2$, fulfilling the properties:

(i) $A\geq c I $, for some $c\in \bR$, if and only if $F_A(x)\geq c$ for all $x\in \sP(\sH_n)$

(ii)  $F_{A^*} = \overline{F_A}$, where the bar denotes the point-wise complex conjugation. In particular   $A=A^*$ if and only if $F_A$ is real.
\end{theorem}
\begin{proof} 
The proof of the first part of (a) is trivial. $F_A$ is continuous and thus bounded, since $\sP(\sH_n)$ is compact. Therefore it belongs to $\cL^2(\sP(\sH_n), d\mu_n)$ as $\mu_n(\sP(\sH_n)) <+\infty$. Concerning (b), we observe that $f(\psi) := F([\psi])$ is a frame function in the sense of definition 
 \ref{ff} due to proposition \ref{propBasis}. If $F \in \cL^2(\sP(\sH_n), d\mu_n)$, then $f \in \cL^2(\bS(\sH_n), d\nu_n)$ in view of the first statement in proposition \ref{propmu}. Thus, whenever $n\geq 3$, we can take advantage of thm \ref{ZAZ},  obtaining  that there is $A\in \gB(\sH_n)$ with $F_A([\psi])= f_A(\psi) = \langle \psi|A\psi\rangle =
\mbox{tr}(\psi \langle \psi| \cdot \rangle A)$ for all $\psi \in \bS(\sH_n)$, namely $F=F_A$.
$A$ is uniquely determined since, as it is simply proved,  in complex Hilbert spaces, if $B: \sH \to \sH$
is linear,  $\langle \psi |B\psi \rangle=0$ for all $\psi \in \bS(\sH)$ then $B=0$. The proof of (c) is evident per direct inspection. Closedness  of $\cF^2(\sH_n)$ for $n\geq 3$ arises form the fact that $\cF^2(\sH_n)$ is a finite dimensional subspace of a Banach space:  The space of quadratic forms on $\sH_n \times \sH_n$ for (b).
\end{proof}

\begin{remark}{\em The statement (b) is false for $n=2$. A simple counterexample is the same as for the classical version of Gleason's theorem. Fix  $p_0 \in \sP(\sH_2)$ and consider the map: $F(p) := \frac{1}{2}\left(1- (1- 2\mbox{tr} (p_0p))^3 \right)$
for $p\in  \sP(\sH_2)$. Passing to the Bloch representation in $\bC^2$, it turns out  evident that this is a positive frame function with $W_F=1$. Next it is simply proved that no $A\in \gB(\sH_2)$ satisfies $F(p) = \mbox{tr}(Ap)$ for all $p\in \sP(\sH_2)$.}
\end{remark}

\subsection{$U(n)$-covariance and (convex) linearity}
\begin{definition}{\em  If ${\cal G} : \sS(\sH_n) \ni \sigma \to g_\sigma$, where the $g_\sigma$ are  maps from $\sP(\sH_n)$
to $\bC$,  we say that $\cal G$, is {\bf $U(n)$-covariant}, if 
$$g_\sigma(\Phi_U(p)) = g_{U^{-1}\sigma U}(p)\quad \mbox{for all $U\in U(n)$, $\sigma \in \sS(\sH_n)$, $p\in \sS(\sH_n)$.}$$}
\end{definition}
\noindent There is a nice interplay between $U(n)$-covariance and frame functions we state and prove in the following theorem which will be a key-tool in the next section.

\begin{theorem}\label{newtheorem}
Assume that $n>2$ for $\sH_n$. \\
{\bf (a)} If ${\cal G} : \sS(\sH_n) \to \cL^2(\sP(\sH_n), \mu_n)$ is a convex-linear and $U(n)$-covariant map, then 
${\cal G}(\sS(\sH_n)) \subset \cF^2(\sH_n)$. \\
{\bf (b)} If ${\cal G}_1 : \gB(\sH_n) \to \cL^2(\sP(\sH_n), \mu_n)$ is a $\bC$-linear map satisfying ${\cal G}_1|_{\sS(\sH_n)}= \cal G$,  with $\cal G$ as in (a), then ${\cal G}_1(\sS(\sH_n)) \subset \cF^2(\sH_n)$.
\end{theorem}

\begin{proof}
(a) Suppose that  $\sigma = \phi\langle \phi| \cdot \rangle $  is a given pure state ad suppose that $\{p_i\}_{i=1,2,\ldots, n}$ is a basis of $\sP(\sH_n)$, so that $p_i = \psi_i \langle \psi_i | \cdot \rangle$.  With a suitable choice of the arbitrary phase in the definition of the $\psi_i$, there are 
$n$ operators $U_i$ such that $U_i \phi = \psi_i$ and $U_i=U_i^*= U_i^{-1}$, 
(lemma \ref{lemmalemme} in the appendix).
 Consequently, taking advantage of the $U(n)$-covariance:
$g_\sigma(p_i) = g_\sigma(U_i\sigma U^*_i) = g_{U^*_i\sigma U_i}(\sigma) = g_{U_i\sigma U^*_i}(\sigma) =
g_{p_i}(\sigma)$.
Exploiting the convex-linearity:
$$n^{-1}  \sum_i g_\sigma(p_i) =
\sum_i n^{-1}g_\sigma(p_i) =  g_{\sum_i n^{-1} p_i}(\sigma) =
 g_{n^{-1}\sum_i  p_i}(\sigma) 
=   g_{n^{-1}I}(\sigma) \:.$$
$U(n)$-covariance implies: $g_{n^{-1}I}(\sigma) = g_{U^{-1}n^{-1}IU}(\sigma) =
g_{n^{-1}I}(\Phi_U(\sigma))$. Since $\Phi$ is transitive on $\sP(\sH_n)$, we conclude that $n^{-1}  \sum_i g_\sigma(p_i)=    g_{n^{-1}I}(q) = c$, for every $q\in \sP(\sH_n)$ and some constant $c\in \bR$.
Next consider a mixed $\sigma \in \sS(\sH_n)$. The found result and convex-linearity of $\cal G$, representing  $\sigma$ with its
spectral decomposition $\sigma = \sum_j q_j \sigma_j$ ($\sigma_j$ being pure),  yield:
$$\sum_i g_\sigma (p_i) =  \sum_i g_{\sum_j q_j\sigma_j} (p_i) =
\sum_i \sum_j q_j g_{\sigma_j}(p_i) 
=    \sum_j q_j \sum_i g_{\sigma_j}(p_i)  =
\sum_j q_j  nc\:.$$
As  the right most side does not depend on the choice of the basis $\{p_i\}_{i=1,2,\ldots, n}$,  $g_\sigma$ must be  a frame function, that belongs to $\cL^2$ by hypotheses.  (b)  of  thm \ref{teoDV2} implies that  $g_\sigma \in \cF^2(\sH_2)$.\\
(b) If $A \in \gB(\sH_n)$,  decompose it  as $A = \frac{1}{2}(A+A^*) + i\frac{1}{2i}(A-A^*)$. Next  decompose the self-adjoint operators $ \frac{1}{2}(A+A^*)$ and $\frac{1}{2i}(A-A^*)$ into linear combinations of pure states $\sigma_k$ exploiting the spectral theorem. Each ${\cal G}_1(\sigma_k)= {\cal G}(\sigma_k)$ belongs to  the linear space $\cF^2$. Linearity of ${\cal G}_1$ concludes the proof.
\end{proof}

\subsection{Trace-integral formulas} 
Frame functions enjoy remarkable properties 
connecting  Hilbert-Schmidt and  $\cL^2(\mu_n)$ scalar products.
These identities were already discovered in  \cite{Gibbons}  for self-adjoint operators, referring to the measure naturally associated with $g$, which  we proved to be proportional to $\mu_n$ in proposition \ref{propmu} above.  Here,  we  establish  them  directly  for $\mu_n$, using the $U(n)$ invariance and dealing with  generally non self-adjoint operators $A,B\in \gB(\sH_n)$.
\begin{theorem}\label{teoFA}
Referring to Theorem \ref{teoDV2},  if  $F_A$ and $F_B$ are frame functions respectively constructed out of $A$ and $B$ in $\gB(\sH_n)$, then:
\beq
\int_{\sP(\sH_n)}  F_A d\mu_n   = \frac{\mbox{tr}(A)}{n} = \frac{W_F}{n}\:,\label{one}
\eeq
\beq
\int_{\sP(\sH_n)} \overline{F_A} F_B d\mu_n   =  \frac{1}{n(n+1)}\left( \mbox{\em tr}(A^*B) + \mbox{\em tr}(A^*)\mbox{\em tr}(B)\right)\:,\label{two}
\eeq
which inverts as: 
\beq
\mbox{\em tr}(A^*B) =  n(n+1)\int_{\sP(\sH_n)} \sp\sp\sp\overline{F_A} F_B d\mu_n  - n^2 \int_{\sP(\sH_n)}\overline{F_A} d\mu_n 
\int_{\sP(\sH_n)}  F_B d\mu_n \:.\label{three}
\eeq
\end{theorem}

\begin{proof}  The second identity in  (\ref{one}) is immediate.
(\ref{three}) arises form (\ref{two}) and 
(\ref{one}) straightforwardly, so we have to prove (\ref{one}) and (\ref{two}) only. Actually (\ref{one}) follows from (\ref{two}) swapping $A$
and $B$ and taking $B=I$. 
Therefore we have to establish (\ref{two}) to conclude.
To this end, we notice that
(\ref{two}) holds true for generic $A,B \in \gB(\sH_n)$ if it is valid for $A$ and $B$ self-adjoint. This result  arises  decomposing $A$ and $B$ in self-adjoint and anti self adjoint part
and exploiting linearity in various points. Therefore 
it is enough proving (\ref{two}) for $A$ and $B$ self-adjoint. Next we observe that, if as before $iu(n)$ is the real vector space of self-adjoint operators:
$iu(n) \times iu(n)  \ni (A,B) \mapsto
(n(n+1))^{-1}\left( \mbox{tr}(AB) + \mbox{\em tr}(A)\mbox{tr}(B)\right)$
is a real scalar product. Similarly, the left-hand side of (\ref{two}), restricted to the real vector space of real frame functions is a real scalar product.  Taking advantage of the  polarization identity, 
we conclude that (\ref{two}) holds when it does for the corresponding norms on the considered real vector spaces:
\beq
\int_{\sP(\sH_n)} F_A^2  d\mu_n   = \frac{1}{n(n+1)}\left( \mbox{tr}(AA) + \mbox{tr}(A)^2\right) \quad \mbox{for $A\in iu(n)$}\:.\label{two2}
\eeq
Let us establish (\ref{two2})  to conclude. We pass from the integration over  $\sP(\sH_n)$ to that over $\bS(\sH_n)$ just replacing $\mu_n$ for $\nu_n$.
 If $\{e_j\}_{j=1,\ldots, n}$ is a Hilbertian basis of $\sH_n$ made of eigenvectors of $A$
such that $Ae_k = \lambda_k e_k$, we can decompose $\psi \in \bS(\sH_n)$
as follows $\psi = \sum_{j} \psi_j e_j$ so that:
$$\int_{\bS(\sH_n)} F_A^2  d\nu_n   = \sum_{i}^n\lambda_i^2 \int_{\bS(\sH_n)}|\psi_i|^4 d\nu_n +
\sum_{i \neq j}^n\lambda_i \lambda_j \int_{\bS(\sH_n)}|\psi_i|^2 |\psi_j|^2 d\nu_n\:. $$
In view of  the $U(n)$ invariance of $\nu_n$ and and the transitive  action of $U(n)$ on $\bS(\sH_n)$,
we conclude that: 
$\int_{\bS(\sH_n)}|\psi_i|^4 d\nu_n(\psi) = a$,
where $a$ does not depend on $i$, on the used Hilbertian basis, and on $A$.
If $\psi, \phi \in \bS(\sH_n)$ are a pair of vectors satisfying $\psi \perp \phi$,
for every choice of $i,j = 1,\ldots,n$ with $i\neq j$, there exist $U_{i,j} \in U(n)$
such that, both verifies $U_{i,j} e_i = \psi$ and $U_{i,j} e_j = \phi$.
The invariance of $\nu_n$ under $U(n)$ thus proves that, for $i \neq j$:
$\int_{\bS(\sH_n)}|\psi_i|^2 |\psi_j|^2d\nu_n(\psi) = b$
where $b$ does not depend on $A$, on the used Hilbertian basis and on the couple 
$i,j = 1,\ldots,n$ provided $i \neq j$.
Summing up:
$$\int_{\bS(\sH_n)} F_A^2  d\nu_n = a\: \mbox{tr}(A^2)  +
b\sum_{i \neq j}^n\lambda_i \lambda_j=  \int_{\bS(\sH_n)} F_A^2  d\nu_n = a\: \mbox{tr}(A^2)  +
b\sum_{i, j}^n\lambda_i \lambda_j - b\: \mbox{tr}(A^2)\:. $$
That is, redefining $d:= a-b$:
\beq \int_{\sP(\sH_n)} F_A^2  d\mu_n = d \:\mbox{tr}(A^2) + b \:(\mbox{tr}(A))^2\:.\label{quasi}\eeq
To determine the constants $d$ and $b$ we first choose $A=I$ obtaining:
$1  = dn + bn^2$.
To grasp another condition,  consider the real vector space of self-adjoint operators $iu(n)$ and 
complete $\frac{I}{\sqrt{n}}$ to a Hilbert-Schmidt-orthonormal basis of $iu(n)$ by adding  self-adjoint operators $T_1, T_2,\ldots, T_{n^2-1}$. Notice that 
$(I|T_k)_2=0$ means $\mbox{tr}(T_k) =0$. Thus, if $p\in \sP(\sH_n)$:
$p = \frac{I}{n} + \sum_k p_k T_k$ with $p_k = \mbox{tr}(pT_k)\in \bR$. The condition $\mbox{tr} (p^2)= 1$ ((2) in proposition \ref{teozero}) is equivalent to $ \sum_k p^2_k =1 - \frac{1}{n}$,
so that
$$\sum_{k=1}^{n^2-1} \int_{\sP(\sH_n)} F_{T_k}(p)^2 d\mu_n(p)
= \int_{\sP(\sH_n)} \sum_{i=1}^{n^2-1}  p_k^2 d\mu_n(p)
=   \left(1- \frac{1}{n} \right)   \int_{\sP(\sH_n)} d\mu_n(p) =  \left(1- \frac{1}{n} \right) \:.$$
 Inserting this result in the left-hand side of (\ref{quasi}):
$$\left(1- \frac{1}{n} \right) = \sum_{i=1}^{n^2-1} d \mbox{tr}(T_iT_i) + b \sum_{i=1}^{n^2-1}(\mbox{tr}(T_i))^2\:,\quad\mbox{i.e.}\quad
\left(1- \frac{1}{n} \right) = \sum_{i=1}^{n^2-1} d  + b \sum_{i=1}^{n^2-1}(0)^2\:.$$
Summing up, we have the pair of equations for $b$ and $d$:
$1-1/n = d (n^2-1)$ and $1  = dn + bn^2$
with solution
$d= b = (n(n+1))^{-1}$ that, inserted in (\ref{quasi}), yields (\ref{two2}).
\end{proof}

\section{Observables and states  in terms of frame functions}\label{sectOFF}

We  know that, assuming   a quantum observable $A \in iu(n)$ 
be associated with a classical-like one $ f_A = F_{\kappa A  + c\: \mbox{tr}(A) I}$  defined in (\ref{CQ}), then  theorem \ref{teoQC} is true independently 
from  the values of the constants $\kappa>0$ and $c\in \bR$.
The first problem we wish to tackle is to study whether there are other possibilities 
to associate quantum observables to classical-like observables preserving the validity  
 theorem \ref{teoQC}.  The second problem we will consider is twofold. On the one hand we want to study if it is possible to associate quantum states $\sigma$
with corresponding classical-like Liouville densities $\rho_\sigma$ in order to satisfy (\ref{trAs0}), possibly with $\rho_\sigma\geq 0$. In this juncture, 
the notion of frame function and
the content of Sect.\ref{Sm} 
 will play a crucial r\^ole.
On the other hand we intend to investigate if all these requirements  give rise to constraints on the values of $\kappa$ and $c$.

\subsection{Observables and states}
In the following we  focus on two maps respectively associating observables with  functions $f_A: \sP(\sH_n) \to \bR$, the {\bf inverse quantization map}:
\beq
{\cal O} : iu(n)  \ni A \mapsto  f_A  \label{cO}\:,
\eeq
and  associating states $\sigma $ with functions $\rho_\sigma:\sP(\sH_n) \to \bR$:
\beq
{\cal S}: \sS(\sH_n) \ni \sigma \mapsto \rho_\sigma\:. \label{cS}
\eeq
 What we intent to do is fixing ${\cal O}$  and ${\cal S}$  by requiring some physically natural constraints, most arising from the thesis of theorem \ref{teoQC} concerning ${\cal O}$ and from the discussion in Sect. \ref{secexp} regarding $\cal S$. We assume the almost K\"ahler  structure $(\omega, g, j)$ on $\sP(\sH_n)$ as in Sect.\ref{secH}, with the constant $\kappa>0$
fixed arbitrarily.\\

\noindent {\bf Requirements on ${\cal O}: iu(n) \ni  A \mapsto  f_A$, with $f_A: \sP(\sH_n) \to \bR$.}\\
{\bf (O1)}  $\cal O$ is injective.\\
{\bf (O2)}  $\cal O$ is $\bR$-linear.\\
{\bf(O3)} If $H\in iu(n)$, then $f_H$ is $C^1$ so that $X_{f_H}$  can be defined. 
A curve $p=p(t) \in \sP(\sH_n)$, $t\in (a,b)$,
satisfies  Hamilton's equation if and only if it satisfies  Schr\"odinger's one:
$$\frac{dp}{dt} = X_{f_H}(p(t)) \quad  \mbox{for}\:\: t \in (a,b) \quad\mbox{is equivalent to}\quad
\frac{dp}{dt} = -i[H, p(t)] \quad    \mbox{for}\:\:  t \in (a,b)\:.$$
{\bf(O4)} $\cal O$ is $U(n)$-covariant.\\
{\bf(O5)} If $A \in iu(n)$ then:   $\min \mbox{sp}(A) \leq f_A(p) \leq \max \mbox{sp}(A) $ for $ p\in \sP(\sH_n)$.\\

 \noindent  
The hypothesis   (O1) simply says that the map $\cal O$ produces a faithful image of the set of quantum observables in terms of classical ones. 
Next  (O2) establishes that ${\cal O}$ also  preserves the elementary structure of the real vector space enjoyed by the set of  quantum observables.
The requirement (O3)  is a key-requirement, it just concerns the interplay of quantum dynamics and Hamiltonian dynamics that we already know to hold when $f_A$ takes the form (\ref{CQ}) in the hypotheses of therem \ref{teoQC}. The requirement (O4), that we know to holds at least when $f_A$ takes the form (\ref{CQ}) in the hypotheses of theorem \ref{teoQC},
is a natural covariance requirement, since the action of $U(n)$ has both classical and quantum significance. 
The requirement (O5) focuses on the values of the observables. It is maybe the most elementary possible  relation between the  values of $A$, the elements of the spectrum, and those of 
 $f_A$, the points in the range. However,  there is no unique way to compare  a continuous set of values with a  discrete one. \\
It is worth stressing that  $f_A$ must be a frame function, $F_{A'}$, in view of (O2) (extended by $\bC$-linearity) and (O4) as a straightforward application of (b) in theorem \ref{newtheorem}. However, it is not so obvious, exploiting  (O2) and (O4) only, to determine the form of the operator $A'$ in terms of $A$ itself.\\

\noindent {\bf Requirements on ${\cal S}: \sS(\sH_n) \ni \sigma \mapsto \rho_\sigma$ with $\rho_\sigma: \sP(\sH_n) \to \bR$.}\\
{\bf(S1)}   If $\sigma \in \sS(\sH_n) $, then  $\rho_\sigma(p) \geq 0$ for $p\in \sP(\sH_n)$.\\
{\bf(S2)}   ${\cal S}$ is convex linear.\\
{\bf(S3)}   With   $\mu_n$  as in theorem \ref{propmu}, if
$\sigma \in \sS(\sH_n) $, then  $\rho_\sigma \in \cL^2(\sP(\sH_n), \mu_n)$ (so that 
 $\rho_\sigma \in \cL^1(\sP(\sH_n), \mu_n)$) and $$ \int_{\sP(\sH_n)} \rho_\sigma d\mu_n =1\:.$$
{\bf(S4)}  $\cal S$ is $U(n)$-covariant.\\
{\bf(S5)} If $\sigma \in \sS(\sH_n) $ and $A \in iu(n)$ then, assuming $f_A\in \cL^2(\sP(\sH_n), \mu_n)$:
$$\mbox{tr}(\sigma A) =  \int_{\sP(\sH_n)} \rho_\sigma f_A d\mu_n\:.$$

\noindent  If we intend to describe quantum expectation values  in terms of classical-like expectation values, the compulsory requirements should be  (S1),(S3) and (S5). 
 The hypotheses   (S2) focuses on the natural  convex structure  of the space of the quantum states requiring  that it is  translated into the  analogue structure for the associated classical-like states. 
(S4) implies in particular  that the  Hamiltonian evolution of $\rho_\sigma$ is equivalent 
to the Schr\"odinger evolution of $\sigma$.
The requirement $f_A\in \cL^2(\sP(\sH_n), \mu_n)$ in (S5) is verified if (O3) holds since, in that case, $|f_A|^2$ is continuous and thus bounded on the compact $\sP(\sH_n)$ and $\mu_n(\sP(\sH_n))<+\infty$. We notice that the map ${\cal S}$ it is not required to be injective, so giving rise to a faithful representation of quantum states in terms of classical-like states. Indeed, we will obtain this result as a consequence of (S2)-(S5). \\

\noindent We are a position to state two of the main results of this paper. 
In particular we prove that, for $n>2$, the densities $\rho_\sigma$ representing quantum states must be frame functions. Beforehand, we  establish that even the classical-like observables $f_A$ are frame functions,  just with  the form (\ref{CQ}) 
and that they exhaust  the whole space $\cF^2(\sH_n)$.

\begin{theorem}\label{main1bis} Consider a quantum system described on $\sH_n$ with $n> 1$. Assume the almost K\"ahler structure $(\omega, g, j)$ on $\sP(\sH_n)$ as in Sect.\ref{secH}, with the constant $\kappa>0$
fixed arbitrarily. The following facts hold concerning the map ${\cal O}: iu(n) \ni A \mapsto f_A$.\\
{\bf (a)} The requirements (O1)-(O4)  are valid if and only if both $\cal O$ has the form  (\ref{CQ}) for some constant $c\in \bR$ (so that, in particular, theorem  \ref{teoQC} holds) and $\kappa +nc \neq 0$.\\
{\bf  (b)}   If the requirements (O1)-(O4)  are valid, then 
 $\cal O$ extends to the whole $\gB(\sH)$ by complex-linearity giving rise to an  injective map that, if $n>2$, satisfies  ${\cal O}(\gB(\sH_n)) =  \cF^2(\sH_n)$.
\end{theorem}

\begin{proof}
(a) If $\cal O$ has the form  (\ref{CQ}) then (O2)-(O4)  are valid. Let us prove the converse.
Assuming (O3), from the definition of Hamiltonian field,  it must be $\omega_p(X_{f_H}, u_A) = \langle df_H, u_A\rangle$,
for every $H\in iu(n)$, $p\in \sP(\sH_n)$ and $u_B = -i[p, B] \in T_p\sP(\sH_n)$.
The definition of $\omega$ and some elementary computations permit to re-write the identity above as
$\langle df_H, -i[p,B]\rangle =\kappa \mbox{tr}(H(-i[p,B]))$. Consider a smooth curve $q=q(s)$  in $\sP(\sH_n)$ such that $q(s_0)=p$ and $\dot{q}(s_0) = -i[p, B]$. The identity above, taking advantage of the linearity of the trace, entails:
$$\frac{d}{ds} f_H(q(s))|_{s=s_0} = \frac{d}{ds} \kappa \mbox{tr}(Hq(s))|_{s=s_0}=
 \kappa \mbox{tr}\left(H \frac{d q}{ds}|_{s=s_0}\right)\:.$$
Since $s_0$ is arbitray, we have found that:
$$\frac{d}{ds} f_H(q(s)) = 
 \kappa \mbox{tr}\left(H \frac{d q}{ds}\right)\:.$$
Integrating in $s$ and swapping the integral with the symbol of trace by linarity, we finally obtain
$f_H(p) = \kappa \mbox{tr}(Hp) + C_H$,  where $p\in\sP(\sH_n)$ is arbitrary.
 The map $H \mapsto C_H = f_H(p) -\kappa \mbox{tr}(Hp)$ must be linear for (O2).
By Riesz' theorem, referring to the Hilbert-Schmidt (real) scalar product we have that there exists $B\in iu(n)$ such that $C_H = \mbox{tr}(BH)$ for all $H\in iu(n)$. (O4)
easily implies that $\mbox{tr}(BUHU^{-1})= \mbox{tr}(BH)$ for al $U\in U(n)$ and $H\in iu(n)$.
Choosing $H= \psi \langle \psi| \cdot \rangle$ with $\psi \in \bS(\sH_n)$ and noticing that $U(n)$ acts transitively on $\bS(\sH_n)$ we conclude that 
$\langle \psi| B \psi \rangle = c$ for some constant $c\in \bR$ and all $\psi \in \bS(\sH_n)$. From polarization identity it easily implies that $B= cI$, so that $f_A= \kappa \mbox{tr}(pA) + c \mbox{tr}(A)$ for all $A \in iu(n)$ as requested.  Let us prove that $\cal O$ is injective if and only if (O1) holds.  Exploiting $\kappa \neq 0$,   and dealing with as done above, one easily sees that $f_A=0$ is equivalent to $A = - c\kappa^{-1} \mbox{tr}(A)  I$. Computing the trace of both sides one immediately sees that this equation has $A=0$ as the unique solution if 
$1+ nc/\kappa \neq 0$, namely   $\kappa + nc \neq 0$. Conversely, if $A + c\kappa^{-1} \mbox{tr}(A)  I =0$ has $A=0$ as unique solution, $I + c\kappa^{-1} \mbox{tr}(I)  I \neq 0$, namely $\kappa + nc \neq 0$.\\
(b)  Assuming (O1)-(O4) and  extending 
$\cal O$ by complex linearity, exactly as before, $\cal O$ turns out to be injective. If $n>2$, the elements of $\cF^2(\sH_n)$ are of the form $F_B$
for every $B\in \gB(\sH_n)$ ((b) thm \ref{teoDV2}). For a fixed $B$, ${\cal O}(A) = F_B$, if $A := \kappa^{-1} B - c [\kappa(\kappa + nc)]^{-1} \mbox{tr}(B) I$,
so that ${\cal O}$ is onto $\cF^2(\sH_n)$.
\end{proof}

\noindent Let us come  to the states, proving that $\rho_\sigma$ is a frame function as well.

\begin{theorem} \label{main2} Consider a quantum system described on $\sH_n$ with $n> 2$. Assume the almost K\"ahler structure $(\omega, g, j)$ on $\sP(\sH_n)$ as in Sect.\ref{secH}, with the constant $\kappa>0$
fixed arbitrarily and the map ${\cal O}: iu(n) \ni A \mapsto f_A$ of the form (\ref{CQ})
for some constant $c\in \bR$. The following fact hold.\\
\noindent {\bf (a)} The requirements  (S2)-(S5) are valid  if and only if both  in the definition  (\ref{CQ}) of ${\cal O}$:
\beq
\kappa  = \kappa \:, \quad   c= c_\kappa\label{A}
\eeq
and ${\cal S}$ associates states $\sigma$ with frame functions of the form:
\beq \rho_\sigma(p) &:=& \kappa'_\kappa  \mbox{\em tr}(\sigma p) + c'_\kappa \label{B}\eeq
where 
\begin{eqnarray}
 c_\kappa &:=&  \frac{1-\kappa}{n} \:,\label{A2}\\
\kappa'_\kappa &:=& \frac{n(n+1)}{\kappa}\:, \quad   c'_\kappa:= \frac{\kappa -(n+1)}{\kappa}  \label{B2}\:.
\end{eqnarray}
The following consequent identities hold
\beq
\kappa + n c_\kappa =1 \quad \mbox{and}\quad  \kappa'_\kappa + n c'_\kappa =n \label{ckappa}\:,
\eeq
and $\cal O$ is injective  due to the first identity above.\\
{\bf  (b)}  If   (S2)-(S5) are true, also ${\cal S}$ is injective.
\end{theorem}

\begin{proof} (a) If  (\ref{A}), (\ref{B}),
(\ref{A2}), (\ref{B2}) are valid, one sees that   (S2)-(S5)  hold true. In particular, $\cal O$ is injective 
because  $\kappa + nc_\kappa =1\neq 0$ and (b) of thm \ref{teoFA} holds.
It remains to prove that (S2)-(S5)   are valid, then 
 (\ref{A}), (\ref{B}),
(\ref{A2}), (\ref{B2}) hold.
We start (for $n>2$) by assuming that a map $\cal S$ verifying  (S2)-(S5) 
and $\cal O$ of the form (\ref{CQ}) with $\kappa>0$ and $c\in \bR$. As the first step we prove that $\rho_\sigma$ is a frame function, next we shall establish its  form.
(S2) and (S4), together with (a) in thm \ref{newtheorem} imply that:
$\rho_\sigma(p) = \mbox{tr}(\sigma' p)$ for all $p\in \sP(\sH_n)$,
where (i) of (c) of  thm \ref{teoDV2} entail that 
$\sigma'$ is some self-adjoint operator associated with the given $\sigma$.
Using the fact that the total integral of $\rho_\sigma$ has value $1$ from (S3), taking (\ref{one}) into account, we find $ \mbox{tr} \sigma' = n$.
Finally (S5) together  the form of $\cal O$ require that the following identity holds true for all  self-adjoint $A\in \gB(\sH_n)$ and $\sigma \in \sS(\sH_n)$:
\beq \mbox{tr}(\sigma A) = \int_{\sP(\sH_n)} \mbox{tr}(\sigma'p)\: (\kappa \mbox{tr}(Ap) + c \mbox{tr} A) \: d\mu_n(p)\:. \label{central}\eeq
The right-hand side can be expanded taking (\ref{two}), (\ref{one}) and $\mbox{tr} \sigma' = n$ into account:
$$\mbox{tr}(\sigma A) =
\frac{\kappa}{n(n+1)} \mbox{tr}(\sigma'A)  + \mbox{tr}(A) \left(\frac{\kappa}{n+1}  + c \right)\:.$$
Consequently, for every $A=A^*$:
$$\mbox{tr}\left(\left(\sigma -\frac{\kappa}{n(n+1)} \sigma'\right)A \right) =
 \mbox{tr}(A) \left(\frac{\kappa}{n+1}  + c \right)\:.$$
Choosing $A=p \in \sS_p(\sH_n)$, arbitrariness of $p$ easily  entails that, for some 
$\beta_\sigma \in \bR$:
$$\sigma -\frac{\kappa}{n(n+1)} \sigma' = \beta_\sigma  I\:,$$
namely, for some constants $\kappa'>0$ and $c' \in \bR$:
$$f_\sigma(p) =  \kappa' \mbox{tr} (\sigma p) + c'\:.$$
Inserting again this expression in (\ref{central}) he have:
$$ \mbox{tr}(\sigma A) = \int ( \kappa' \mbox{tr} (\sigma p) + c')\: (\kappa \mbox{tr}(Ap) + c \mbox{tr} A) \: d\mu_n(p)\:.$$
Finally, using again (\ref{one}), (\ref{two}) and $\mu_n(\sP(\sH_n))=1$ we obtain:
$$\left(1- \frac{\kappa \kappa'}{n(n+1)} \right)  \mbox{tr} (\sigma A) +
\mbox{tr}(A) \left(\frac{\kappa \kappa'}{n(n+1)} + cc' + \frac{ck'}{n} + \frac{c'k}{n}\right)=0$$
that has to hold for all $A\in iu(n)$ and $\sigma \in \sS(\sH_n)$
and for some $\kappa,\kappa' >0$ and $c,c' \in \bR$. Arbitrariness of $A$ and $\sigma$ easily lead to the first two requirements of the following triple:
$$1- \frac{\kappa \kappa'}{n(n+1)}  = 0 \:, \quad 
\frac{\kappa \kappa'}{n(n+1)} + cc' + \frac{ck'}{n} + \frac{c'k}{n}=0\:, \quad 
1 = \frac{\kappa'}{n} + c'\:,$$
the  third requirement immediately arises from (S3) using (\ref{one}).
This system can completely  be solved  parametrizing $\kappa, \kappa', c$ in terms of  $c'$ with $c'<1$ in order to verify the requirement $\kappa>0$ in the definition of 
$f_A$.
Finally,  parametrizing the solutions 
in terms of $\kappa$: $\kappa, c_\kappa, k'_\kappa, c'_\kappa$
we have  (\ref{A2}), (\ref{B2}).\\
(b)  If $\rho_\sigma= \rho_{\sigma'}$, exploiting $\kappa'_\kappa \neq 0$, one has $\mbox{tr}((\sigma-\sigma')p)=0$ for every $p\in \sP(\sH_n)$. Namely $\langle \psi | (\sigma-\sigma') \psi \rangle =0$ for every 
$\psi \in \sH_n$. Polarization leads to $\sigma-\sigma'=0$.
\end{proof}

\begin{remark}$\null$\\{\em
{\bf (1)} Assuming (S5), $\rho_\sigma$ can be represented by  a frame function for  expectation values of   observables
$f_A \in \cF^2(\sH_n)$ as an immediate consequence of Riesz' theorem, noticing that $ \cF^2(\sP(\sH))$ is a closed subspace of $L^2(\sP(\sH_n), \mu_n)$.
However this observation nothing says about the nature of $\rho_\sigma$ when  the expectation values are computed for  classical-like observables  $f$ with components  in  $ \cF^2(\sP(\sH))^\perp$.\\
{\bf (2)} The pair of identities (\ref{ckappa}) respectively imply that (1) a quantum observable $A = aI$,
with $a\in \bR$ constant, corresponds to $f_A$ costantly attaining the value $a$ and (2) that the
completely unpolarized state $\sigma = n^{-1}I$ gives rise to the classical trivial Liouville density
$\rho_\sigma = 1$ costantly.}
\end{remark}
 
\subsection{Characterization of 
classical-like observables representing quantum observables}  The proved theorem, as a by product, yields a characterization of 
classical-like observables representing quantum ones when $n>2$. It is well known (see \cite{BH01})  that 
these observables $f$ are exactly  those whose Hamiltonian fields $X_f$ are $g$-Killing fields.  Focusing instead on the relation of observables and states, another characterization is the following.

\begin{proposition}\label{propchat} For $n>2$, let  $\cal O$ and $\cal S$ be as in (a) of theorem \ref{main2}.\\ A map $f:  \sP(\sH_n) \to \bR$ in $\cL^2( \sP(\sH_n), \mu_n)$ verifies  $f={\cal O}(A)$ for some $A \in iu(n)$ if and only if there are constants $a,b \in \bR$ with $a\neq 0$ and
\beq \int_{\sP(\sH_n)} \rho_{p_0}(p) f(p) d\mu_n(p) = af(p_0) +b \quad \mbox{for every $p_0 \in \sS_p(\sH_n)$.} \label{char}\eeq
\end{proposition}

\begin{proof}
If $f=f_A$ one has immediately:
$$\int_{\sP(\sH_n)} \rho_{p_0}(p) f(p) d\mu_n(p) = \mbox{tr}(p_0A) =
 \kappa^{-1}  f(p_0) - \kappa^{-1} c_\kappa \mbox{tr}(A) \quad \mbox{for every $p_0 \in \sS_p(\sH_n)$.}$$ 
Conversely, assume that (\ref{char}) holds for a map $f:  \sP(\sH_n) \to \bR$ in $\cL^2( \sP(\sH_n), \mu_n)$. If $\{p_i\}_{i=1,\ldots, n}$ is a basis of $ \sP(\sH_n)$ one has:
$$n^{-1}a\left(\sum_i f(p_i)\right) + b = \int_{\sP(\sH_n)} \sum_i n^{-1}\rho_{p_i}(p) f(p) d\mu_n(p)
= \int_{\sP(\sH_n)}\rho_{ \sum_i n^{-1} p_i}(p) f(p) d\mu_n(p)$$
$$= \int_{\sP(\sH_n)}\rho_{ n^{-1}I}(p) f(p) d\mu_n(p) = 
 \int_{\sP(\sH_n)} f(p) d\mu_n(p)\:.$$
So that $$\sum_i f(p_i) = \frac{n}{a} \int_{\sP(\sH_n)} f(p) d\mu_n(p) -\frac{n b}{a}$$ that does not depend on the choice of the basis $\{p_i\}_{i=1,\ldots, n}$. In view of (b) in thm \ref{teoDV2}, $f$ is a real frame function. Due to (d) of thm \ref{main2}, $f= {\cal O}(A)$ per some $A\in iu(n)$.
\end{proof}

\noindent With the choice, $\kappa =1$, the proposition above specializes to $a=1$ and $b=0$. This gives rise to a suggestive interpretation of the Liouville densities of pure states:
$$\int_{\sP(\sH_n)} \rho_{p_0}(p) f_A(p) d\mu_n(p) = f_A(p_0)\:.$$
If $\kappa=1$, a map $f:  \sP(\sH_n) \to \bR$ in $\cL^2( \sP(\sH_n), \mu_n)$ can be written as $f=f_A$ for some $A\in iu(n)$ if and only if 
$f$ ``sees'' the density $\rho_{p_0}$ of any pure state $p_0$ as a Dirac delta localized at $p_0$ itself.

\begin{remark}$\null$\\{\em
Since we now have a complete characterization of classical-like observables as scalar functions on quantum phase space regardless of the notion of linear operators, it is convenient stating an inverse formula to associate self-adjoint operators with  functions. Using trace-integral formulas (\ref{one}) and (\ref{two}),  the following formula arises for any $A\in iu(n)$ \cite{D}:
$$A=\int_{\sP(\sH_n)} f_A(p) \gO(p) d\mu_n(p)\:.$$
The operator-valued function $\gO:\sP(\sH_n)\rightarrow \gB(\sH_n)$ is defined as:
\beq\label{re-q}
\gO(p) := \frac{(n+1)}{\kappa}p-\left(\frac{n+1-\kappa}{\kappa\,n}\right )\bI.
\eeq
Therefore, if $f:\sP(\sH_n)\rightarrow \bR$ is a classical-like observable,  the associated self-adjoint operator  is obtained smearing $f$ against  the kernel (\ref{re-q}).
}

\end{remark}

\subsection{Bounds on attained values}  To conclude our analysis, let us focus on the validity of (S1) and (O5).
As stated in the following theorem, they cannot hold simultaneously.

\begin{theorem}\label{teobound}
For $n>1$, with $\cal O$ and $\cal S$ defined in accordance  with  (\ref{A}), (\ref{B}),
(\ref{A2}), (\ref{B2}) for some $\kappa>0$, the following facts are valid.\\
{\bf (a)} (S1) holds if and only if $\kappa \in [n+1, +\infty)$. \\
{\bf (b)}  (O5) holds if  and only if $\kappa \in (0,1]$ whereas,  in the general case 
$\kappa >0$ one has:
\begin{align}
\min f_A  &=  \min \mbox{\em sp}(A) + c_\kappa (\mbox{\em tr}(A) -n  \min \mbox{\em sp}(A))\:,\label{m} \\
\max f_A &=  \max \mbox{\em sp}(A) + c_\kappa (\mbox{\em tr}(A) - n \max \mbox{\em sp}(A))\:, \label{M}
\end{align}
and furthermore,  for $A= iu(n)$:
\begin{eqnarray}
&||f_A||_\infty \leq  (1 + 2n|c_\kappa|) ||A|| \quad &\mbox{if $\kappa \in [n+1, +\infty)$}\label{est2}\:,\\
&||f_A||_\infty \leq  ||A|| \quad &\mbox{if $\kappa \in(0,1]$}\label{est}\:,
\end{eqnarray}
where $\leq$ can be replaced by $=$ if $\kappa=1$.
\end{theorem}

\begin{proof}
(a)  if  (\ref{A}), (\ref{B}),
(\ref{A2}), (\ref{B2}) are valid, (S1) holds if and only if
 $\kappa'_\kappa >0$ and $c'_\kappa\geq 0$
(notice that $\sigma \geq 0$ for a state by hypotheses and there are states  with $\mbox{tr}(p\sigma)=0$ for some $p \in \sP(\sH_n)$). From  (\ref{B2}),  $\kappa'_\kappa >0$ and $c'_\kappa\geq 0$ are equivalent to
$\kappa \in [n+1, +\infty)$.  \\
(b)  We known that, since $A$ is self-adjoint and $\sH_n$ is finite dimensional,
  $\min \mbox{sp}(A) = \min_{||\psi||=1} \langle \psi| A\psi\rangle = \min_{p \in \sP(\sH_n)} \mbox{tr}(pA)$.  Therefore
$\min f_A = \kappa \min \mbox{sp}(A) + c_\kappa \mbox{tr}(A)$. Using 
$\kappa + n c_\kappa =1$ we immediately have (\ref{m}). The proof of 
(\ref{M}) is analogous.  Next notice that $\mbox{tr}(A) -n  \min \mbox{sp}(A) \geq 0$ and 
$\mbox{tr}(A) -n  \max \mbox{sp}(A) \leq 0$ so that (\ref{m})-(\ref{M})  imply (O1)
if and only if $c_\kappa \geq 0$, namely
$\kappa \in (0,1]$. 
The proof of the remaining estimates easily follows 
using an analogous procedure, noticing that $\kappa >0$ and  exploiting the inequalities $|\mbox{tr}(A)| \leq \mbox{tr}|A| \leq n||A||$ which arises  from $||A||= \max\{ |\lambda 
|\:|\: \lambda \in \mbox{sp}(A)\}$ and  $\max_{p\in \sP(\sH_n)} | \mbox{tr}(pA)| =  ||A||$. The latter implies the validity of the last statement in (b)
out of  the fact that $c_\kappa =0$ if $\kappa=1$.
\end{proof}

\begin{remark} {\em $\null$\\
{\bf (1)} For $n=2$, the implication "if" in (a) of theorem \ref{main2}
 holds true.\\
{\bf (2)}  The choice $\kappa = n+1$ implies $\kappa'_\kappa =n$, $c'_\kappa=0$, so that: $\rho_\sigma(p) =n F_\sigma(p) =n \mbox{tr}(p\sigma)$
 are positive classical densities with the most elementary form allowed by our hypotheses. 
This form may further be  simplified changing the normalization of the measure. Leaving $\cal O$ unchanged, one may indeed redefine $\mu_n \to \mu'_n := n\mu_n$ and $\rho_\sigma \to \rho_\sigma' := n^{-1}\rho_\sigma$ to obtain $\rho'_\sigma = F_\sigma$ exactly, preserving (S1)-(S5) with $\rho_\sigma'$ in place of $\rho_\sigma$, but $\mu'_n(\sP(\sH_n))=n$.\\
{\bf (3)} Gibbons' choice corresponds to $\kappa =1$ so that  $c_\kappa=0$ and thus: $f_A(p)= F_A(p) = \mbox{tr}(pA)$,
defining the  simplest relation between quantum and classical-like observables.\\
{\bf (4)}  The established theorem shows in particular that the pair of  requirements  (S1)
and (O5) cannot hold simultaneously. As long as our goal is  describing quantum physics by means of the Hamiltonian formalism, it seems preferable  assuming the validity of the former ($\kappa \in [n+1,+\infty)$), dropping the latter. Otherwise  $\rho_\sigma$ would not be represented in terms of a {\em positive} Liouville probability density. Sticking to this choice, 
we can also use $\rho_\sigma$ to evaluate expectation values for observables 
that are not of quantum nature, differently form what instead happens if $\kappa = 1$.  Assuming $(S1)$, the failure of (O5) seems however to remain  annoying. Actually,  it is not so strong  as it could seems at first glance, since as 
already stressed, there is no unique way to compare  a continuous set of reals (the range of $f_A$) with  a discrete set of real numbers (the spectrum of $A$) and  the only physically sensible  comparison relies upon the identity (\ref{trAs0}) (with $\mu_n$ in place of $m$) that is satisfied.
In particular, this identity assures that all elements of $\mbox{sp}(A)$ are always obtained as  expectation values of $f_A$ with respect to suitable classical-like states\footnote{If $\kappa = 1$,  the elements $\mbox{sp}(A)$ coincides to the singular values of $f_A$ (i.e. $df_A(p) =0$ iff $f_A(p) \in \mbox{sp}(A)$) as one easily proves.}: If $a \in \mbox{sp}(A)$, 
picking out $p_a \in \sS_p(\sH_n)$ such that $p_a = \psi_a \langle \psi_a| \cdot \rangle$, where $\psi_a$ is a normalized eigenvector of $A$ with eigenvalue $a$, one has $\bE_{\rho_{\sigma_a}}(f_A)= \langle A \rangle_{\sigma_a}= a$.
The apparent  difficulties arise  only when  one is dealing with {\em intrinsically classical} states (including classical sharp states) {\em and quantum} observables. For instance, if $\rho \neq \rho_\sigma$ for every $\sigma \in \sS(\sH_n)$, it could happen that $A\geq 0$ but $\bE_\rho(f_A) <0$  (it is impossible if $\rho = \rho_\sigma$ since
$\bE_\rho(f_A)= \langle A\rangle_\sigma \geq 0$ for  (\ref{trAs0})).}
\end{remark}

\subsection{$C^*$-algebra of quantum observables in terms of  frame functions}
In this section we assume to work with $\cal O$ of the form
(\ref{CQ}), when holding  (\ref{A}), (\ref{B}),
(\ref{A2}), (\ref{B2}). The observables of the systems we are considering are the self-adjoint elements of $\gB(\sH_n)$. Considering also complex combinations of observables we recover the whole $C^*$-algebra $\gB(\sH_n)$. The map ${\cal O}$, defined with respect to a choice of $\kappa>0$,  can be extended by linearity to a map indicated with the same symbol:
$${\cal O} : \gB(\sH_n) \ni A \mapsto f_A := \kappa F_A + c_\kappa \mbox{tr}(A) \in 
\cF^2(\sH_n)\:.$$
From (d) in theorem \ref{main2}, this map turns out to be an isomorphism of complex vector spaces with the further property that
$${\cal O}^{-1}\left(  \overline{f}\right) = \left({\cal O}^{-1}( f)\right)^* \quad   \mbox{for all $f \in \cF^2(\sH_n)$.}$$
Obviously $\cal O$ can be used to induce on $\cL^2(\sP(\sH_n), \mu_n)$ a structure of $^*$-algebra, defining a (distributive and  associative) $^*$-algebra product:
\beq f \star g := {\cal O}\left( {\cal  O}^{-1}(f) {\cal  O}^{-1}(g) \right)\label{prod} \quad \mbox{for all $f,g \in \cF^2(\sH_n)$.}\eeq
With this product $\cF^2(\sH_n)$ becomes a $C^*$-algebra with unit, given by the constantly function $1$,  with involution given by the standard complex conjugation and norm:
\beq |||f||| :=  ||{\cal O}^{-1} (f) || \quad   \mbox{for all $f \in \cF^2(\sH_n)$,}\label{norm}\eeq 
where the norm in the right-hand side is  the $C^*$-norm of  $\gB(\sH_n)$.
With these definitions, ${\cal O}$ turns out  to be a $C^*$-algebra isomorphism.
 The proofs are straightforward.\\
An intriguing issue is whether the $C^*$-algebra norm and products can be recast using the 
geometric structure already present on $\sP(\sH_n)$.  
Concerning the norm, we observe that it is enough to know the explicit expression for the case of $f$ real, the general case then arises form that and the $C^*$-algebra property  $||a^*a||= ||a||^2$, once one has an explicit expression for the product $\star$, that we will obtain shortly.  As a matter of fact  $||| f||| = \sqrt{||| \overline{f} \star f|||}$. Focus on $f \in \cF^2(\sP(\sH_n))$
real, so that we can write $f=\kappa F_A + c_\kappa \mbox{tr}(A)$ for some $A \in iu(n)$. Since $A$ is self-adjoint,
$$||F_A||_\infty = \sup_{p\in \sP(\sH_n)} |\mbox{tr}(pA)| =   \sup_{\psi\in \bS(\sH_n)} |\langle \psi |A \psi \rangle| = ||A||  = ||| f|||\:.$$ 
As a consequence, taking advantage from (\ref{one}), from the explicit expression of $f_A$ (\ref{cO}),
and exploiting $\kappa + nc_\kappa =1$, we immediately have a geometric expression for $|||f|||$.
\begin{proposition}
If $n>2$ and  $f \in \cF^2(\sH_n)$ is real, 
 referring to the 
 the $C^*$-algebra  norm in (\ref{norm}) (everything defined for a choice of $\kappa>0$) we have:
\beq ||| f|||  = \frac{1}{\kappa}\left| \left| f - (1-\kappa) \int_{\sP(\sH_n)} f d\mu_n  \right|\right|_\infty \:.\label{norm}
\eeq
\end{proposition}

\begin{remark} {\em Even dropping the requirement $f \in \cF^2(\sH_n)$ and assuming, more generally, $f\in C^0(\sP(\sH_n))$, the right-hand side of (\ref{norm}) defines a norm. The same holds true 
if working in $L^\infty(C^0(\sP(\sH_n)), \mu_n)$ and interpreting $||\:\:||_\infty$ as the natural norm referred to the essential supremum computed with respect to $\mu_n$. The proofs are straightforward.}
\end{remark}

\noindent We write down two cases explicitly. 
The  case $\kappa =n+1$, with $\mu'_n := n\mu_n$:
$$||| f|||  = \frac{1}{n+1}\left| \left| f +  \int_{\sP(\sH_n)} f d\mu'_n  \right|\right|_\infty\:,\quad \mbox{for every real  $f \in \cF^2(\sH_n)$,}$$
and the case considered by Gibbons, $\kappa =1$:
$$||| f|||  =\left| \left| f\right|\right|_\infty\:,\quad \mbox{for every real  $f \in \cF^2(\sH_n)$.}$$
Let us finally pass to the product $\star$, stating a corresponding theorem.

\begin{theorem}\label{teolast}
Let  $n>2$ and  $f, g \in \cF^2(\sH_n)$. 
If  $G_p: T_p^*\sP(\sH_n) \times T_p^*\sP(\sH_n) \to \bR$ denotes the scalar 
product on $1$-forms  canonically  induced  by the metric $g$ on $\sP(\sH_n)$,
referring to the $C^*$-algebra product
  in (\ref{prod}), we have:
$$  f \star g = \frac{i}{2}\{f, g\} +\frac{1}{2} G(df,dg) + \frac{fg}{\kappa}  + \frac{1-\kappa}{\kappa} \left(\frac{n+1}{\kappa} \int_{\sP(\sH_n)} \sp \sp\sp\sp f g d\mu_n  - f\sp  \int_{\sP(\sH_n)}\sp \sp\sp\sp g d\mu_n-g  \sp \int_{\sP(\sH_n)}\sp \sp\sp\sp f d\mu_n \right) $$
\beq 
+ \frac{1-\kappa}{\kappa^2}(\kappa - (n+1)) \int_{\sP(\sH_n)} \sp \sp\sp\sp f d\mu_n\int_{\sP(\sH_n)} \sp \sp\sp\sp  g d\mu_n
 \label{prod2}\eeq
with, as usual, 
$\kappa>0$. In particular, for  $\kappa  =n+1$ and defining $\mu'_n := n\mu_n$, one has
$$  f \star g = \frac{i}{2}\{f, g\} + \frac{1}{2} G(df,dg) +   \frac{1}{n+1} \left(  
fg  -\int_{\sP(\sH_n)} \spa \sp \sp\sp\sp f g \: d\mu'_n + f\sp  \int_{\sP(\sH_n)}\spa \sp \sp\sp\sp g \: d\mu'_n+g  \sp \int_{\sP(\sH_n)}\spa \sp \sp\sp\sp f \: d\mu'_n \right)  $$
 and, for $\kappa =1$,
$$ f \star g = \frac{i}{2}\{f, g\} + \frac{1}{2} G(df,dg) + fg\:.$$
\end{theorem}

\begin{proof}
First of all we notice that, the following identity holds,
$g(X_f,X_g) = G(df,dg)$,
that immediately follows from $g( X_f, j\cdot ) = \omega(X_h,  \cdot)  =
 df$, $jj=-I$ and $g(ju, jv ) = g(u,v)$. So we replace 
$G(df,dg)$ for $g(X_f,X_g)$ in the following.
Define $f:= f_A$ and $g:= f_B$. With this choice $f\star g = f_{AB}=
\kappa F_{AB} + c_\kappa \mbox{tr}(AB)$.
Therefore:
$$f_{AB}(p) = \kappa \mbox{tr}(pAB) + c_\kappa \mbox{tr}(AB) =  \frac{\kappa}{2}\kappa \mbox{tr}(p[A,B]) 
+ \frac{\kappa}{2} \mbox{tr} (p (AB+BA)) +  c_\kappa \mbox{tr}(AB)\:.$$ 
A straightfrward computation proves that
$$ \mbox{tr} (p (AB+BA)) = - \mbox{tr}(p [p,A][p,B]) - \mbox{tr}(p [p,B][p,A]) + 2 \mbox{tr}(pA)\mbox{tr}(pB)\:.$$
Reminding the definition of $\omega$, $\{\cdot, \cdot\}$, $X_h$ and $g$ presented in Sect.\ref{secH},  and putting all together we find:
$$ (f \star g)(p)=  f_{AB}(p) = \frac{i}{2}\{f, g\} + \frac{1}{2} g(X_f,X_g) + \kappa \mbox{tr}(pA)\mbox{tr}(pB) +
 c_\kappa \mbox{tr}(AB)\:.$$
From $\kappa + n c_\kappa=1$ and using (\ref{one}), one easily finds
$$\int_{\sP(\sH_n)} f_A d\mu = \kappa \int_{\sP(\sH_n)} F_A d\mu + c_\kappa
 \mbox{tr}(A)= \frac{1}{n} \mbox{tr}(A)  = \int_{\sP(\sH_n)} F_A d\mu\:,$$
and a similar result for $B$ in place of $A$.
Using (\ref{three}) and the definition of $f_A$ ($f_B$) in terms of $F_A$ ($F_B$):
$$ f \star g = \frac{i}{2}\{f, g\} + \frac{1}{2} g(X_f,X_g) + \frac{1}{\kappa} \left(f - nc_\kappa \int_{\sP(\sH_n)} \sp\sp\sp f d\mu_n \right) \left(g - nc_\kappa \int_{\sP(\sH_n)} \sp\sp\sp g d\mu_n \right)$$
$$ -c_\kappa n^2 \int_{\sP(\sH_n)} \sp\sp\sp f d\mu_n \int_{\sP(\sH_n)} \sp\sp\sp g d\mu_n 
+ \frac{c_\kappa n(n+1)}{\kappa^2}
\int_{\sP(\sH_n)} \left(f - nc_\kappa \int_{\sP(\sH_n)} \sp\sp\sp f d\mu_n \right) \left(g - nc_\kappa \int_{\sP(\sH_n)} \sp\sp\sp g d\mu_n \right) d\mu_n \:.$$
Using again $\kappa + n c_\kappa=1$ we obtain (\ref{prod2}).
\end{proof}

\begin{remark}$\null$\\
{\em {\bf (1)} As already noticed in \cite{BH01}, for $\kappa=1$ it turns out that the squared standard deviation 
$(\Delta A)_\psi$ (where $p = \psi \langle \psi| \cdot\rangle$) coincides to $\frac{1}{2}G_p(df_A, df_A)$. This allows one to write down  a geometrical formulation of Heisenberg inequality. For other choices of $\kappa$ it is still possible, but the expression is more complicated.\\
{\bf (2)} A formula similar to (\ref{prod}) for $n=2$ (stated on the $2$-dimensional Bloch sphere) and for the case $\kappa =1$ is mentioned in \cite{lieJordan2}.\\
{\bf (3)}  From (\ref{prod2}), the structure of Lie-Jordan Banach algebra  \cite{FFIM13,lieJordan} of $\cF^2(\sH_n)$
shows up evidently. The Lie commutator is just $\{\cdot, \cdot\}$
whereas  the Jordan product reads:
$$f \circ g := \frac{1}{2} G(df,dg) + \frac{fg}{\kappa}  + \frac{1-\kappa}{\kappa} \left(\frac{n+1}{\kappa} \int_{\sP(\sH_n)} \sp \sp\sp\sp f g d\mu_n  - f\sp  \int_{\sP(\sH_n)}\sp \sp\sp\sp g d\mu_n-g  \sp \int_{\sP(\sH_n)}\sp \sp\sp\sp f d\mu_n \right)$$
$$+ \frac{1-\kappa}{\kappa^2}(\kappa - (n+1)) \int_{\sP(\sH_n)} \sp \sp\sp\sp f d\mu_n\int_{\sP(\sH_n)} \sp \sp\sp\sp  g d\mu_n\:.$$}
\end{remark}

 \section{Conclusions and open issues}
This paper discussed some issues regarding the interplay of standard and geometric formulation of finite dimensional quantum mechanics, working  in the projective space. All  the analysis was based  upon  the properties of the  $\cL^2(\mu_n)$ frame functions, focusing on the r\^ole of  $U(n)$ covariance in particular.
 The problem of the positiveness of Liouville densities associated with quantum states was tackled,  establishing  that the geometric  formalism,  in view of the existence  of a one-parameter class of  natural K\"ahlerian structures on $\sP(\sH_n)$,  permits  to fix several  physically safe solutions. 
 A new characterization of classical-like observables  describing quantum observables was  discussed, together with a geometric description of the $C^*$-algebra structure (decomposed as a Banach Lie-Jordan algebra) of the set of quantum observables in terms of the  K\"ahlerian structure.\\
It seems  worth remarking that the results  remain affected by a  free parameter, $\kappa>0$, that could not  be fixed  out of  our physical requirements.
 A possibility to get rid of this remaining freedom could arise from the description of compound systems. Compound systems cannot completely described in classical-like terms because a problem pops up from scratch. In classical physics the phase space of a system composed by two  subsystems is the Cartesian product of the space of phases of each subsystem.
If one tries to extend the approach of this paper, the phase space of the composed system must be instead the projective space of the {\em tensor product} of the  Hilbert spaces  of the subsystems.  That is  much larger than the Cartesian product of the projective spaces. The Cartesian product is actually embedded in this larger manifold   by means of the well known  {\em Segre embedding}. A huge literature exists on this topic.  It would be interesting to analyse this issue with the help of the technology of frame functions. \\
Another direction to investigate is, obviously, the infinite dimensional case. 
Barring evident technical problems  to be fixed, a very difficult issue is the generalization of the measure $\mu_n$ on the complex projective space of an infinite dimensional Hilbert space.

\appendix
\section{Proof of some results}

{\bf Proof of proposition \ref{teoquat}}.
\begin{proof}
(1)  First of all, notice that  the three norms $||\cdot||$, $||\cdot||_1$, $ ||\cdot||_2$ are topologically equivalent since $\gB(\sH_n) = \gB_1(\sH_n) =\gB_2(\sH_n)$ are  finite dimensional   normed  spaces with respect to the corresponding norms. \\
(2) $\sS(\sH_n)$ is compact since it is closed and bounded, with respect to the norm $||\cdot||_1$, in a finite dimensional normed space.
Since $\sS(\sH_n)$ is compact and the zero operator $0 \not \in \sS(\sH_n)$, the continuous functions $\sS(\sH_n) \ni \sigma \mapsto d(0, \sigma) = ||\sigma||$, 
$\sS(\sH_n) \ni \sigma \mapsto d_1(0, \sigma) = ||\sigma||_1$, $\sS(\sH_n) \ni \sigma \mapsto d_2(0, \sigma)= ||\sigma||_2$
must have strictly positive minima (and maxima). For $d_1$ everything is trivial. Let us pass to consider $d$ and $d_2$. Using the fact that the $n$ eigenvalues $q_k$ of $\sigma \in \sS(\sH_n)$ verify 
both $q_k \in [0,1]$ and $\sum_{k=1}^n  q_k =1$ one sees that $\sum_{k=1}^n 
q^2_k \geq \frac{n}{n^2}$, and $1/n$  is indeed the least possible value. All that is equivalent to say that $||\frac{1}{n}I ||_2\leq ||\sigma||_2$ where $\frac{1}{n}I $ is an admitted state.  Again with the constraints 
$q_k \in [0,1]$ and $\sum_{k=1}^n  q_k =1$,  the maximum of the eigenvalues $q_k = |q_k|$ must be greater than $1/n$, that is equivalent to say $||\sigma|| \geq ||\frac{1}{n}I ||$. Concerning maxima,  $q_k \in [0,1]$ and $\sum_{k=1}^n  q_k =1$ imply, varying 
$\sigma \in \sS(\sH_n)$: 
$1 \geq \sum_{k=1}^n q(\sigma)^2_k$ and  
$\max\{q(\sigma)_k\:|\: k=1,\ldots, n\:, \sigma \in \sS(\sH_n)\} =1$
determining the maximum of both  $\sigma \mapsto ||\sigma||_2$ and $\sigma \mapsto  ||\sigma||$ since the value $1$ of the norms  is attained on pure states. \\
 (3).  We view  $\sS(\sH_n)$ a subset  of the topological   space $T$ of self-adjoint operators $A$ on $\sH_n$ with $\mbox{tr}(A)=1$  endowed with  the topology induced by $\gB(\sH_n)$.\\
First of all, notice that $\sS(\sH_n) \supset \partial \sS(\sH_n)$ because the former is closed with respect to the said topology, so $\sS(\sH_n)$
is the disjoint union of $\partial \sS(\sH_n)$ and $Int(\sS(\sH_n))$.
Let $\sigma$ be an  element of $\sS(\sH_n)$.  First suppose that $\mbox{dim}(\mbox{Ran}(\sigma)) = n$
we want to show that $\sigma \in Int(\sS(\sH_n))$, that is, there is an open set  $O\subset  T$ containing
 $\sigma$ and  such that $\sigma' \in O$ entails $\sigma' \in \sS(\sH_n)$. 
To this end, let us define
$m := \min\{\:  \langle \psi| \sigma \psi \rangle \:\:|\:\: ||\psi|| =1\}\:.$
$m$ is real  since $\sigma= \sigma^*$ and $m>0$, because: (1) all eigenvalues of $\sigma$ are strictly positive (since $\sigma\geq 0$ and $\mbox{dim}(\mbox{Ran}(\sigma))=n$),  (2) $\psi \mapsto \langle \psi | \sigma \psi\rangle$ is continuous and (3) the set of vectors $\psi$ with 
$||\psi||=1$ is compact because $\mbox{dim}(\sH_n)=n < +\infty$. Next, if $\sigma' = \sigma'^*\in \gB(\sH_n)$ verifies
$||\sigma- \sigma'|| < m/2$, one has:
$$\frac{m}{2} \leq  ||\sigma- \sigma'|| = \sup \{  \:|\langle \psi |(\sigma - \sigma') \psi \rangle|\:\:\:|\:\:\: ||\psi||=1\}$$
so that:
$\langle \psi |\sigma' \psi \rangle =  \langle \psi |\sigma' \psi \rangle  - \langle \psi |\sigma \psi \rangle
+  \langle \psi |\sigma \psi \rangle \geq  - \frac{m}{2} + m = \frac{m}{2}>0 \quad \mbox{for $||\psi||=1$.}$
Consequently: $\sigma' \geq 0$.  Summarizing, if $B_{m/2}(\sigma)$ denotes the open ball 
in $\gB(\sH_n)$ centred on $\sigma$ with radius $m/2$,  
$O:=T \cap B_{m/2}(\sigma)$ is open in $T$ by definition and $\sigma'\in O$ verifies $\sigma' =\sigma'^*$,  $\mbox{tr} \sigma =1$ and, as  we have proved, $\sigma'\geq 0$.  In other words, for $\sigma \in \sS(\sH_n)$, 
$\mbox{dim}(\mbox{Ran}(\sigma)) = n$ implies  $\sigma \in Int(\sS(\sH_n))$.\\
We pass to the other case  for $\sigma \in \sS(\sH_n)$. We suppose that  $\mbox{dim}(\mbox{Ran}(\sigma)) < n$  and we want to show that 
$\sigma \in \partial \sS(\sH_n)$.
 $\mbox{dim}(\mbox{Ran}(\sigma)) < n$ implies  $det(\sigma)=0$. Thus all eigenvalues are non-negative and one at least  vanishes. Let $\psi \in Ker(\sigma)$. The operators, for $n=1,2,\ldots$:
 $$\sigma_n := \left(1 + \frac{1}{n} \right) \sigma - \frac{1}{n} \psi\langle \psi| \cdot \rangle\:, $$
  are self-adjoint with $\mbox{tr}(\sigma_n)= 1$  so that they stay in $T$. Furthermore  they 
verify $\sigma_n \to \sigma$ for $n\to +\infty$, but $\sigma_n \not \in \sS(\sH_n)$ because $\sigma_n$
has the negative eigenvalue $-\frac{1}{n}$. So $\sigma \in \partial \sS(\sH_n)$. In particular, since 
a pure state is a one-dimensional orthogonal projector, it verifies  $\mbox{dim}(\mbox{Ran}(\sigma)) =1 < n$ and thus 
$\sigma \in \partial \sS(\sH_n)$. If $n=2$ this is the only possible case for an element  $\sigma \in \partial \sS(\sH_n)$.
However, if $n>2$, also elements of $\sS(\sH_n)$ with $\mbox{dim}(\mbox{Ran}(\sigma)) \leq n-1$
belong to $\partial \sS(\sH_n)$. 
\end{proof}

\noindent {\bf Proof of proposition \ref{profgeo}}.
\begin{proof} 
 (a) $\bS(\sH_n)$ is a  {\em real $(2n-1)$-dimensional embedded submanifold} of $\sH_n$. Let us  sketch how it happens.  If $(z_{01}, \ldots, z_{0n}) \in \bS(\sH_n)$, there is an open (in $\sH_n$)  neighbourhood $O$ of that point such that  for every  $(z_1, \ldots, z_n)\in  O' :=  \bS(\sH_n) \cap O$ there is a component, say $z_k = x_k + i y_k$, (the same for all points of $O$) such that either  $x_k$ or $y_k$ can be written  as a smooth function  of the remaining components  $z_h$, when  $(z_1, \ldots, z_n)\in O'$. This procedure define a natural local chart on $\bS(\sH_n)$  with domain $O'$. Collecting all these charts, that are mutually smoothly compatible,  one obtains a smooth differentiable structure on $\bS(\sH_n)$ making it a  real $(2n-1)$-dimensional embedded submanifold  of $\sH_n$.\\
(b) Similarly  $\sP(\sH_n)$ can be equipped with a real $(2n-2)$-dimensional smooth  manifold structure.
Consider $(z_{01}, \ldots, z_{0n}) \in [\psi_0] \in \sP(\sH_n)$.
 At least one of the components $z_{0j}$
cannot vanish, say $z_{0h}$. By continuity this fact is valid  in an open  neighbourhood $V$ of $(z_{01}, \ldots, z_{0n}) \in \bS(\sH_n)$. In that neighbourhood  the set of $n-1$ ratios $z_j/z_h$ with $j\neq h$ determine a point on $\sP(\sH_n)$ biunivocally. These $n-1$ ratios vary in an  open neighborhood 
$V' := \pi(V) \subset \sP(\sH_n)$  of $[\psi_0]$ when the components $(z_{1}, \ldots, z_{n})$ range in $V$.
Decomposing each of these ratios into real and imaginary part, we obtain a real local chart on $V' \subset \sP(\sH_n)$ with $2n-2$ real coordinates.  Collecting all these local charts, that are mutually smoothly compatible,  one obtains a smooth differentiable structure on $\sP(\sH_n)$, making it a $2n-2$-dimensional real smooth manifold.
With the said structures, the canonical projection  $\pi  : \bS(\sH_n) \to \sP(\sH_n)$ 
becomes a smooth submersion
and the transitive action (\ref{ta}) of $U(n)$ on $\sP(\sH_n)$ turns out to be smooth as one easily proves.
\end{proof}

\noindent{\bf Proof of proposition \ref{propvett}}.
\begin{proof} The action (\ref{ta}) is  transitive and smooth so, on the one hand
$\sP(\sH_n)$ is diffemorphic to the quotient $U(n)/G_p$, where $G_p \subset U(n)$ is the isotropy (closed Lie sub) group of $p\in \sP(\sH_n)$ and on the other hand  the projection $\Pi_p : U(n) \ni U \mapsto Up U^{-1}\in \sP(\sH_n)$ is a submersion \cite{W,Oneill} and thus $d\Pi_p|_{U=I} : u(n) \to T_\sigma \sP(\sH_n)$ is surjective.  The thesis is true because $d\Pi_p(B)|_{U=I} = [B, p]$ for every $B \in u(n)$ and $u(n)$ is the real vector space  of anti-self adjoint  in $\gB(\sH)$. 
\end{proof}

\noindent{\bf Proof of theorem \ref{teoQC}}.
\begin{proof} Regarding (\ref{XfA}), consider a smooth curve $\bR \ni t \mapsto p(t) \in \sP(\sH_n)$ such that 
$\dot{p}(0)= v = -i[B_v,p]$, (\ref{CQ}) implies:
$$\langle df_A(p), v\rangle = \kappa \frac{d}{dt}|_{t=0}\mbox{tr} (p(t) A) = -i \kappa \mbox{tr} [[B_v, p],A]
= -i  \kappa \mbox{tr}(p [A,B]) = \omega_p(-i[A,p],v)\:.$$
Since it must also hold $\omega_p(X_{f_A}(p),v) =\langle df_A(p), v\rangle$ we conclude that
$\omega_p(X_{f_A}(p)+i[A,p],v)=0$ for every $v\in T_p\sP(\sH_n)$. As $\omega_p$ is non-degenerate, (\ref{XfA}) follows.\\
(a) In view of (\ref{XfA}), Hamilton equation $\frac{dp}{dt} = X_{f_H}(p(t))$ is the same as Schr\"odinger equation $\frac{dp}{dt} = -i[H, p(t)]$. Now the final statement is obvious form (\ref{CQ}) and the cyclic property of the trace:
$$f_A(p(t)) = f_A(e^{-itH}pe^{itH})= \kappa \mbox{tr}\left(e^{-itH}pe^{itH}A \right)
+ c \mbox{tr}(A) = \kappa \mbox{tr}\left(pe^{itH}Ae^{-itH} \right)
+ c \mbox{tr}(e^{-itH}e^{itH} A) $$
$$ = 
\kappa \mbox{tr}\left(pe^{itH}Ae^{-itH} \right)
+ c \mbox{tr}(pe^{itH}Ae^{-itH} ) = f_{e^{itH}Ae^{-itH}}(p)\:.$$
(b) The first statement immediately arises using $\{f_A,f_B\}:= \omega(X_{f_A},X_{f_B})$
and (\ref{XfA}) noticing that $\mbox{tr}\left(-i[A,B]\right)=0$. Next observe that, since $\sH_n$ has finite dimension and so no problems with domains arise, $A$ is a constant of motion with respect to the Hamiltonian $B=H$ iff $[A,H]=0$. This is equivalent to say $\langle\psi| [A,H] \psi\rangle =0$ for all $\psi \in \bS(\sH_n)$, that is $\mbox{tr}(p [A,H])=0$ for all $p\in \sP(\sH_n)$, that is $f_{-i[A,B]}
- \mbox{tr}\left(-i[A,B]\right)=0$. In view of the very identity (\ref{ppc}), that is equivalent to say (where $\cH = f_H$)
$\{f_{A}, \cH\}=0$, that is eventually equivalent to say that $f_A$ is a constant of motion in Hamiltonian formulation.\\
(c) The first part of (c) is an evident  consequence of the fact that referring to   (\ref{ta}):
$\Phi^\star_U\omega = \omega$  and $\Phi^\star_U g = g$ for all $U\in U(n)$.
In other words  $\omega_p(u,v)$ and $g_p(u,v)$ are invariant  if replacing $p$, $A_v$, $A_u$
for $Up U^{-1}$, $UA_vU^{-1}$, $UA_uU^{-1}$ simultaneously as one checks immediately.
The last statement is immediately proved by direct inspection.
\end{proof}

\noindent {\bf Proof of proposition \ref{propmu}}
\begin{proof} Henceforth $\cB(X)$ denotes the Borel $\sigma$-algebra on the topological space $X$. 
If $\mu_n$ exists, the requirement
$\int_{\sP(\sH_n)} f d\mu_n  = \int_{\bS(\sH_n)} f \circ \pi  d\nu_n$ entails that,
for $f:= \chi_E$, it holds:
\beq \mu_n(E) = \nu_n(\pi ^{-1}(E)) \quad \mbox{for every $E \in \cB(\sP(\sH_n))$\:.}\label{mn}\eeq
That relation proves that, if $\mu_n$ exists, it is uniquely determined by $\nu_n$. Let us pass to the existence issue.
Since $\pi $ is continuous,  is Borel-measurable and thus $\pi ^{-1}(E) \in \cB(\bS(\sH_n))$
if $E \in \cB(\sP(\sH_n))$. Since the other requirements are trivially verified, (\ref{mn}) defines, in fact,  a positive Borel measure on $\sP(\sH_n)$. That measure fulfils \beq
\mbox{$f\circ \pi   \in \cL^1(\bS(\sH_n), \nu_n)$} \quad\mbox{if}
\quad \mbox{$f \in \cL^1(\sP(\sH_n), \mu_n)$,}\quad \mbox{and}\quad 
\int_{\sP(\sH_n)} f d\mu_n  = \int_{\bS(\sH_n)} f \circ \pi  \:\: d\nu_n\:. \nonumber
\eeq directly from the definition of integral and $\mu_n(\sP(\sH_n)) = \nu_n(\pi ^{-1}(\sP(\sH_n))) = \nu_n (\bS(\sH_n))=1$. 
 $\mu_n$ is regular because $\sP(\sH_n)$ is compact it being  the image of  the compact set $\bS(\sH_n)$ under the continuous map $\pi $
 with finite measure \cite{rudin} and this regularity 
results also applies the the Liouville measure and the Riemannian one.
Concerning the invariance under the action of $PU(n)$, it arises from that of $\mu_n$ under $U(n)$:
$$ \mu_n(E) =  \nu_n (\pi ^{-1}(E)) = \nu_n (U\pi ^{-1}(E)) =  \nu_n(\pi ^{-1}(U E  U^{-1})) = \mu_n(U E  U^{-1}) \quad \mbox{if $U \in U(n)$.}$$
(a)  $\sP_n(\sH_n)$ is homeomorphic to the quotient of compact groups $U(n)/H$
where $H$ is the  isotropy group of any point of $\sP_n(\sH_n)$, since $H$ is closed in the compact group $U(n)$, it is compact as well. Thus there is a non-vanishing  $U(n)$ left invariant 
positive regular Borel measure on $\sP_n(\sH_n)$, uniquely determined by the volume of $\sP_n(\sH_n)$ (Chapter 4 of \cite{BR}). That measure must thus coincide with $\mu_n$ up to a strictly positive multiplicative constant. (b) and (c) the Liouville measure and the Riemannian measure are  non-vanishing  $U(n)$ left invariant 
positive regular Borel measure on $\sP_n(\sH_n)$, because both $\omega$ and $g$ are $U(n)$
invariant. Therefore they have to coincide with $\mu_n$ up to a strictly positive multiplicative constant.
\end{proof}

\begin{lemma}\label{lemmalemme}
If $\sH$ is a complex Hilbert space and $\phi, \psi  \in \bS(\sH)$, then
 there exists $U \in \gB(\sH)$ such that  $U\phi = \alpha \psi $ for some  $\alpha \in \bC$, $|\alpha|=1$, and $U= U^* =U^{-1}$.
\end{lemma}

\begin{proof}
If $\phi$ and $\psi$ are linearly dependent, choosing $\alpha$ such that 
$\phi=\alpha \psi$, we can define $U:= I$. In the other case, let $\sK$ be the 
closed subspace spanned by $\phi$ and $\psi$. It is enough to find 
$V: \sK \to \sK$ with $V= V^*=V^{-1}$ and $V\phi = \alpha \psi$ for some $\alpha$
with $|\alpha|=1$. If such a $V$ exists, the wanted $U$ can be defined ad 
$U:= V \oplus I$ referring to the orthogonal decomposition $\sH = \sK \oplus \sK^\perp$.  Fixing an orthonormal basis in $\sK$ given by $\phi, \phi_1$, the problem can be tackled in $\bC^2$. With the Hilbert-space isomorphism 
from $\sK$ to $\bC^2$, $\phi$ corresponds to $(1,0)^t$ and $\psi$ with $(a,b)^t$
where $|a|^2+ |b|^2=1$. We can choose $\alpha$ such that 
$\alpha \psi$ corresponds to  $(c,d)$ with $c>0$, $d\in \bC$  and 
$c^2 + |d|^2=1$. To conclude, we only  need to find a complex $2\times 2$ matrix $M$
with $M= \overline{M}^t = M^{-1}$ and $M(1,0)^t = (c, d)^t$. The operator $V$ corresponds to $M$ through the identification of $\sK$ and $\bC^2$ we have previously introduced. The wanted $M$ is  just the following one:
$$M:= \left[
\begin{array}{cc}
c & \overline{d} \\ d & -c
\end{array}\right]\:.$$

\end{proof}
\section*{Acknowledgements}
The authors are grateful to S. Chiossi, G. Marmo and E. Pagani for useful discussions.


\begin{thebibliography}{References}

\bibitem[AS95]{AS}  A. Ashtekar and T. A. Schilling, {\em Geometry of quantum mechanics}, AIP Conference Proceedings, {\bf 342}, 471-478 (1995).


\bibitem[BH01]{BH01} D. C. Brody, L. P. Hughston, {\em Geometric quantum mechanics},  Journal of Geometry and Physics {\bf 38} 19-53  (2001)

\bibitem[BSS04]{BSS04} A. Benvegn\`u, N.Sansonetto and M.Spera,
{\em Remarks on geometric quantum mechanics},
Journal of Geometry and Physics {\bf 51}  229-243 (2004)

\bibitem[BR00]{BR} A. O. Barut and  R.  Raczka, {\em Theory of Group Representations and Applications}.
World Scientific, Singapore,  2000

\bibitem[CIMM09]{lieJordan2} 	
J. F. Carinena,  A. Ibort, G. Marmo, G. Morandi, {\em Geometrical description of algebraic structures: Applications to Quantum Mechanics},
AIP Conf.Proc. {\bf 1130} 47-59  (2009)

\bibitem[CMP94]{CMP} R. Cirelli, A. Mani\`a, L.Pizzocchero, {\em A functional representation for noncommutative C* algebras},
Rev. Math. Phys. {\bf 6}, 675-697 (1994)


\bibitem[CLM83]{CLM} R. Cirelli, P. Lanzavecchia, and A. Mani\`a, {\em Normal pure states of the von neumann algebra
of bounded operators as kahler manifold} J. Phys. A: Mathematical and General,
{\bf 16}, 3829 (1983)

\bibitem[Dvu92]{Dvu} A. Dvure$\check c$enskij,  \emph{Gleason's theorem and its applications} (Kluwer academic publishers, 1992).


\bibitem[FFIM13]{FFIM13}
F. Falceto, L . Ferro,  A. Ibort and G.  Marmo, {\em  Reduction of Lie-Jordan Banach algebras and quantum states }
 J. Phys. A: Math. Theor. {\bf 46},  015201 (2013)


\bibitem[FFMP13]{lieJordan}  P. Facchi,  L. Ferro, G.  Marmo and
S. Pascazio, {\em Defining quantumness via the Jordan product}, Preprint 
arXiv:1309.4635v1

\bibitem[GCM05]{GCM2005} J. Grabowski, M. Ku\'s and G. Marmo, {\em Geometry of quantum systems: density states and
entanglement}, J. Phys. A: Math. Gen. {\bf 38},  10217-10244 (2005) 


\bibitem[Gib92]{Gibbons} G. W. Gibbons. {\em Typical states and density matrices},
Journal of Geometry and Physics {\bf 8}  147-162  (1992)



\bibitem[Gle57]{Gleason} A. M. Gleason, \emph{Measures on the closed subspaces of a Hilbert space}, Journal of Mathematics and Mechanics, Vol.6, No.6, 885-893 (1957).


\bibitem[Kib79]{Kibble} T.  W.  B. Kibble,  {\em Geometrization of Quantum Mechanics}, Commun. Math. Phys. {\bf 65}, 189-201 (1979)


\bibitem[Mo13]{Mo} V. Moretti, {\em Spectral Theory and Quantum Mechanics}, Springer , Berlin, 2013

\bibitem[MP13]{DV1}  V.  Moretti and D.  Pastorello, {\em  Generalized Complex Spherical Harmonics, Frame Functions, and Gleason Theorem},
Ann. Henri Poincar\'e {\bf 14},1435-1443  (2013).

\bibitem[Pa15]{D} D. Pastorello. {\em A geometric Hamiltonian description of composite quantum systems and quantum entanglement}, Int. J. Geom. Methods Mod. Phys., {\bf 12}, 1550069 (2015)


\bibitem[Mac51]{Mackey} G. W. Mackey, \emph{Induced representations of locally compact groups I}. The Annals of Mathematics, Second Series, Vol.55, No.1, 101-139 1951.

\bibitem[ONe83]{Oneill} B. O'Neill, {\em Semi-Riemannian geometry}, Academic Press,  San Diego (1983)


\bibitem[Rud86]{rudin} W. Rudin, \emph{Real and complex analysis} (McGrow-Hill Book Co. 1986).


\bibitem[War83]{W} F. W. Warner, {\em  Foundations of Differentiable Manifolds and Lie Groups}, Springer, Berlin,  1983
\end{thebibliography}
\end{document}